%% file: main.tex
\pdfoutput=1
\documentclass{article}

\usepackage[english]{babel}
\usepackage{csquotes}

\usepackage[a4paper,top=2cm,bottom=2cm,left=3cm,right=3cm,marginparwidth=1.75cm]{geometry}

\usepackage{amsmath}
\usepackage{amsthm}
\usepackage{amssymb}
\usepackage{bbm}
\usepackage[usenames,dvipsnames]{xcolor}
\usepackage{hyperref}

\hypersetup{pdfauthor={Behrens Arpino Kivva Zdeborova},pdftitle={(Dis)assortative Partitions on Random Regular Graphs},%
            colorlinks, linktocpage=true, pdfstartpage=3, pdfstartview=FitV,%
    breaklinks=true, pdfpagemode=UseNone, pageanchor=true, pdfpagemode=UseOutlines,%
    plainpages=false, bookmarksnumbered, bookmarksopen=true, bookmarksopenlevel=1,%
    hypertexnames=true, pdfhighlight=/O,%
    urlcolor=orange, linkcolor=blue, citecolor=black, %pagecolor=RoyalBlue,%
        }

\usepackage[style=alphabetic]{biblatex}
\usepackage{authblk}
\usepackage{graphicx}
\usepackage{graphbox}
\usepackage{booktabs}
\usepackage{algorithm2e}
\usepackage{caption}
\usepackage{subcaption}

\RestyleAlgo{ruled} 

\newcommand{\minH}{H}

\newcommand{\minh}{h}
\newcommand{\ind}{\mathbbm{1}}

\newcommand{\smat}[4]{[\begin{smallmatrix}
#1 & #2\\
#3 & #4
\end{smallmatrix}]}

\DeclareMathOperator*{\argmin}{arg\,min}
\newtheorem{conj}{Conjecture}
\newtheorem{thm}{Theorem}

\newcommand{\vcenteredinclude}[1]{\begingroup
\setbox0=\hbox{#1}%
\parbox{\wd0}{\box0}\endgroup}

\title{(Dis)assortative Partitions on Random Regular Graphs}
\author[1]{Freya Behrens}

\author[2]{Gabriel Arpino}

\author[3]{Yaroslav Kivva}

\author[1]{Lenka Zdeborová}

\affil[1]{{\small Statistical Physics of Computation Lab, École Polytechnique Fédérale de Lausanne (EPFL). }}

\affil[2]{{\small Department of Engineering, University of Cambridge. Invenia Labs.}}

\affil[3]{{\small Business Analytics Lab, École Polytechnique Fédérale de Lausanne (EPFL). }}

\date{}

\begin{document}
\maketitle

\begin{abstract}
We study the problem of assortative and disassortative partitions on random \texorpdfstring{$d$}{d}-regular graphs. Nodes in the graph are partitioned into two non-empty groups.
In the assortative partition every node requires at least \texorpdfstring{$H$}{H} of their neighbors to be in their own group.
In the disassortative partition they require less than \texorpdfstring{$H$}{H} neighbors to be in their own group. Using the cavity method based on analysis of the Belief Propagation algorithm we establish for which combinations of parameters \texorpdfstring{$(d,H)$}{(d,H)} these partitions exist with high probability and for which they do not. 
For \texorpdfstring{$H>\lceil \frac{d}{2} \rceil $}{H > ceil(d/2)} we establish that the structure of solutions to the assortative partition problems corresponds to the so-called frozen-1RSB.
This entails a conjecture of algorithmic hardness of finding these partitions efficiently. For \texorpdfstring{$H \le \lceil \frac{d}{2} \rceil $}{H <= ceil(d/2)} we argue that the assortative partition problem is algorithmically easy on average for all \texorpdfstring{$d$}{d}. Further we provide arguments about asymptotic equivalence between the assortative partition problem and the disassortative one, going trough a close relation to the problem of single-spin-flip-stable states in spin glasses. In the context of spin glasses, our results on algorithmic hardness imply a conjecture that gapped single spin flip stable states are hard to find which may be a universal reason behind the observation that physical dynamics in glassy systems display convergence to marginal stability. 
\end{abstract}

\section{Introduction}
\begin{figure}[b]
    \centering
    \includegraphics[width=0.45\textwidth]{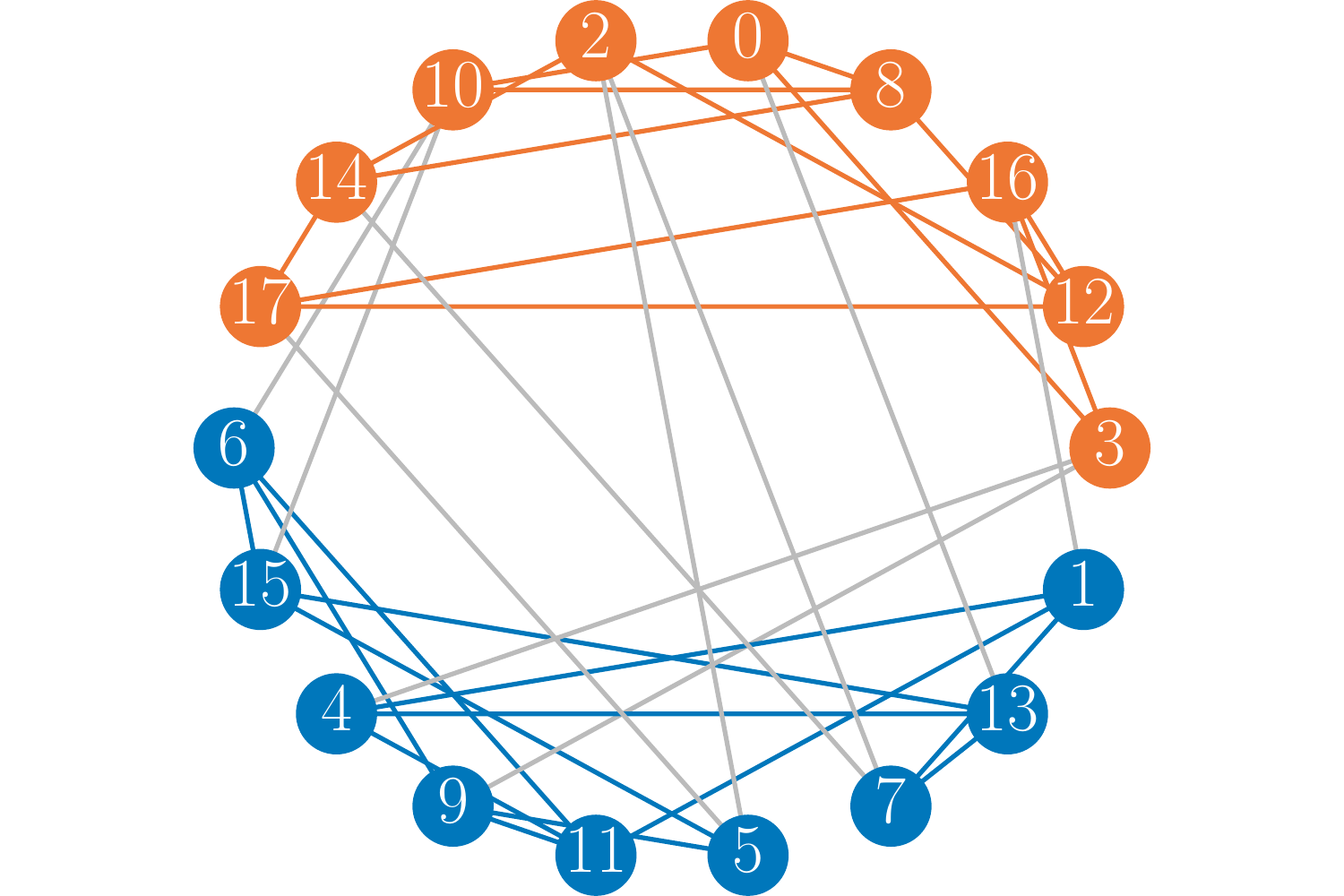}
    \includegraphics[width=0.45\textwidth]{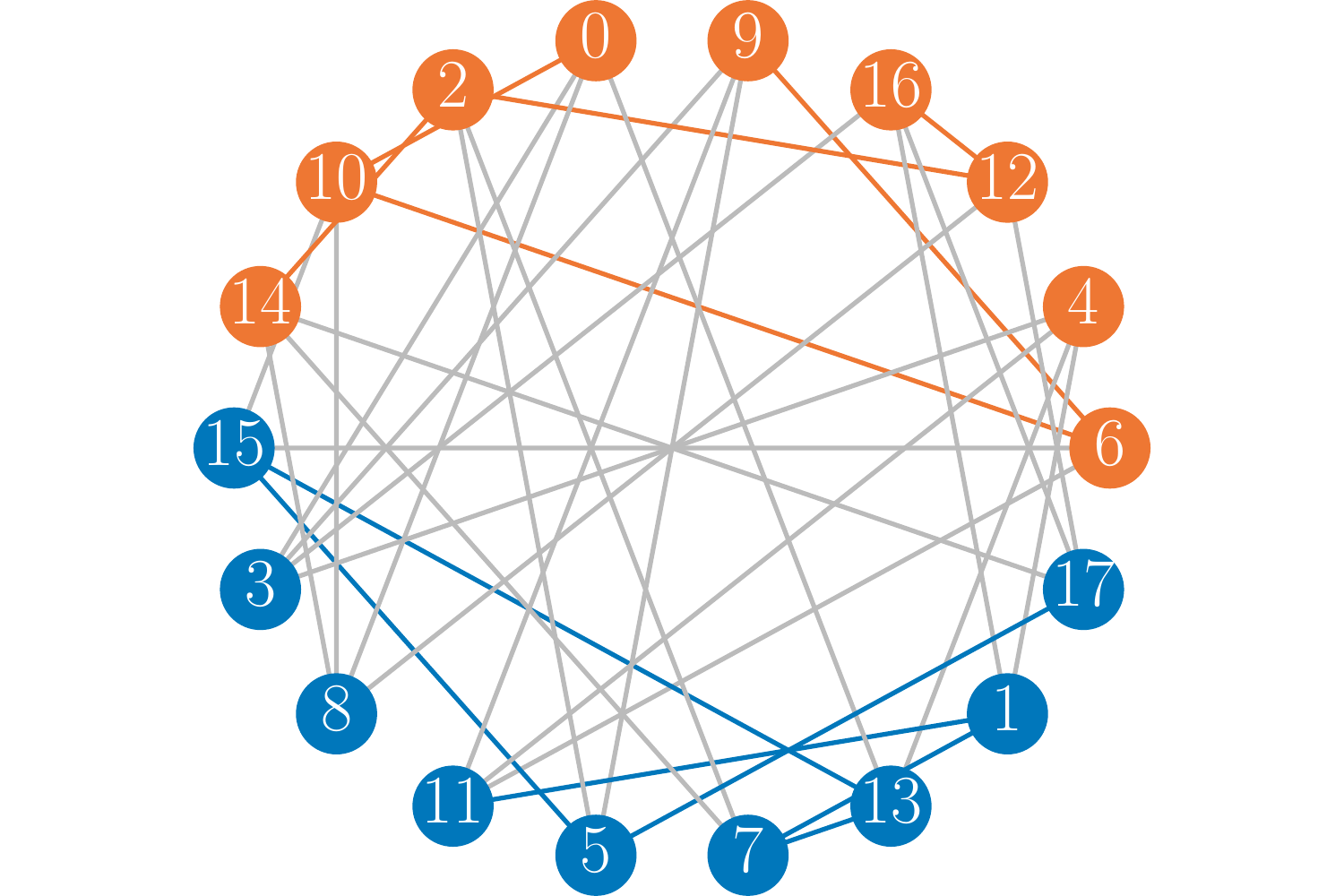}
    \caption{\textit{A $4$-regular graph with $18$ nodes.} \textit{(Left)} Example of a $2$-assortative bipartition, where every nodes has at least two of their neighbors in their own group. \textit{(Right)} Example of a $3$-disassortative bipartition, where every node has less than 3 of their neighbors in their own group, i.e. at least two in the opposite group.}
    \label{fig:toy-example}
\end{figure}

Let $G(V,E)$ be a simple undirected graph with $n$ nodes $i = \{1,...,n\}$ and edges $(i,j) \in E \subseteq V^2$.
We consider a partition of the nodes $V$ into two disjoint non-empty sets $\langle V_-,V_+ \rangle$.
Let the  \textit{magnetization} of the partition be $m = \frac{|V_+|-|V_-|}{|V|} \in [-1,+1]$ and a \textit{bisection} be a partition with $m=0$.
For the parameter $H$ we call $i \in V$ an \textit{assortative} node, if at least $H$ of its neighbors are in $i$'s own group.
Conversely, if less than $H$ neighbors are in its own group, the node $i$ is called \textit{disassortative}.

In this paper we study partitions under which the nodes in $V$ are either all assortative or all disassortative with respect to a given $H$, as exemplified in Figure~\ref{fig:toy-example}.
We call these two complementary problems $H$\textit{-assortative partition} and $H$\textit{-disassortative partition}.
When as an additional constraint the magnetization is fixed to some value $m$ we consider it as $(H,m)$-(dis)assortative partition.
For the special case of balanced groups with a fixed $m=0$ we simply refer to it as $H$-(dis)assortative \textit{bipartition}.

An early illustrative example of the assortative partition problem was posed in the context of game theory.
A group of people with alliance relations is divided into two subgroups where everyone wants to have more allies in its own group than in the other \cite{kristiansenAlliancesGraphs2004}.
If every person has exactly $d$ ally relations, a $\lceil\frac{d}{2}\rceil$-assortative partition of an ally graph is a Nash Equilibrium of the game.
In this case, every person has at least half of their allies in their own group and consequently nobody has an incentive to change their group unilaterally.

Similar constrained partition problems have emerged in different practical contexts, for example in social choice, community detection or process control (see \cite{bazganSatisfactoryGraphPartition2010} for a survey).
This broad variety of origins created an abundance of terminology in the field.
Therefore, results on assortative partitions can be found under the name of \textit{cohesive subsets}~\cite{morrisContagion2000}, \textit{satisfactory partitions}~\cite{bazganSatisfactoryGraphPartition2010}, \textit{friendly partitions}~\cite{ferberFriendlyBisectionsRandom2021a}, \textit{generalized matching cuts}~\cite{gomesFindingCutsBounded2021} and \textit{local minimum bisection}~\cite{alaouiLocalAlgorithmsMaximum2021}.
Likewise, disassortative partitions relate to \textit{co-satisfactory partitions}~\cite{bazganSatisfactoryGraphPartition2010}, \textit{unfriendly partitions}~ \cite{ferberFriendlyBisectionsRandom2021a} and \textit{local maximum cut}~\cite{alaouiLocalAlgorithmsMaximum2021}.

In statistical physics, the problem of $\lceil\frac{d}{2}\rceil$-assortative partitions has been investigated due to its relation to \textit{single-spin flip stable states} in the ferromagnetic Ising model. The $\lceil\frac{d}{2}\rceil$-disassortative partitions correspond to single-spin flip stable states in the anti-ferromagnetic Ising model. The spin glass version of the single-spin flip stable states was also widely considered \cite{brayMetastableStatesInternal1981}. The assortative partition problem also has analogies with the stable states in Hopfield networks \cite{hopfieldNeuralNetworksPhysical1982}.

Many canonical combinatorial problems, such as graph coloring or independent sets, have been studied on random graphs and this line of work often unveiled an intriguing geometric and algorithmic behavior. 
In particular, analysing the typical cases of these problems is interesting to reach a better understanding of heuristic algorithms.
To this end, it is customary to analyse ensembles of random graphs such as Erdős–Rényi or random $d$-regular graphs.
The latter are drawn uniformly from the set of graphs with size $n$ where every node has degree $d$ and $dn$ is even. In the limit where $n\to \infty$ and $d$ is a constant these graphs are locally tree-like,  which is a key property underlying the applicability of a range of probabilistic methods \cite{mezard2009information}. Understanding the structure of the solution space of (dis)assortative partitions on random $d$-regular graphs, associated algorithmic consequences, and the relation between the two problems is the focus of this work.

\section{Related Work}

\paragraph{Existence of (dis)assortative partitions and computational complexity}
There is a pletho\-ra of work that considers the algorithmic complexity of finding (dis)assortative (bi)partitions or even just deciding whether they exist on general non-random graphs.
We present an overview of these results in Figure~\ref{fig:related-work}.

For general, i.e. non-random, $d$-regular graphs and small $H \leq \lceil \frac{d}{2} \rceil - 1$ the question whether assortative partitions always exist was answered positively by \cite{stiebitzDecomposingGraphsDegree1996}: They can also be found in polynomial time, albeit the same does not hold true for a given fixed magnetization~\cite{bazganExistenceDeterminationSatisfactory2003}.
Naturally, this result directly transfers to random $d$-regular graphs.

For larger $H = \lceil \frac{d}{2} \rceil$ a long standing unresolved conjecture states that such a $H$-assortative partition always exists on $d$-regular graphs of size at least $n_d$~\cite{banInternalPartitionsRegular2013}.
For small $d<7$ the conjecture was proven by \cite{shafiqueSatisfactoryPartitioningGraphs2002,banInternalPartitionsRegular2013}, except for $d=5$.
Moreover, it was later shown that in the asymptotic limit of $d$, for even-degree regular graphs there is always a $\lceil\frac{d}{2}\rceil$-assortative partition \cite{linialAsymptoticallyAlmostEvery2017}.
In fact, for larger $\lceil \frac{d}{2} \rceil + 1 \leq H \leq d-1$ deciding  whether a $H$-assortative partition exists on $d$-regular graphs was shown to be NP-hard when $d$ is even by \cite{gomesFindingCutsBounded2021}.

On the other hand, when in addition to $H= \lceil \frac{d}{2} \rceil$ the magnetization is fixed to some $m$ it can be shown that such a partition does not exist on infinitely many $d$-regular graphs.
For $m=0$ and $d=3$ a simple counterexample is the wheel graph, a cycle graph where all nodes are additionally connected to their diametrically opposed neighbor.
In this graph a $\lceil\frac{d}{2}\rceil$-assortative bipartition does not exist for arbitrarily large~$n$ \cite{bazganSatisfactoryGraphPartition2010}.
When no satisfying partition exists for $m=0$ and $H=\lceil \frac{d}{2} \rceil$, there are no polynomial time approximation schemes for the minimum number of disassortative nodes of a balanced partition or the maximum unbalance of a assortative partition \cite{bazganSatisfactoryGraphPartition2010}.

Another relevant line of work on the existence of assortative partitions examines which subgraphs must be forbidden in order for an $\lceil\frac{d}{2} \rceil$-assortative partition to exist \cite{stiebitzDecomposingGraphsDegree1996,gerberAlgorithmicApproachSatisfactory2000,maDecomposingC4freeGraphs2019,liuConjectureSchweserStiebitz2021}.
For example, a sufficient criterion is for cycles of length 4 to be excluded from the graph \cite{maDecomposingC4freeGraphs2019}.
Since random regular graphs are locally tree-like and contain few short cycles \cite{mckayShortCyclesRandom2004}, this result already indicates what we will find later, that most random graphs have $\lceil d/2 \rceil$-partitions that are easy to find.

For example, there exists an algorithm for $\lceil\frac{d}{2} \rceil$-assortative partitions on dense ER graphs that provably finds them with high probability \cite{ferberFriendlyBisectionsRandom2021a}.
This algorithm uses a greedy node swapping operation to iteratively approach a solution.
Similar algorithms using node flipping or swapping operations were analysed in the framework of the complexity class of polynomial local search (PLS) \cite{christopoulosOverviewWhatWe2004}.
This class includes problems for which one local step in the solution space, i.e. finding direct neighbors and computing their costs, takes polynomial time in the worst case.
Examples from this class are local Maximum Cut \cite{angelLocalMaxcutSmoothed2017} and finding stable configurations in Hopfield networks  \cite{schafferSimpleLocalSearch1991} which relate closely to weighted versions of our problem as discussed in the remainder of this section. Negative results about local search algorithms were proven using the so-called overlap gap property in a closely related class of max-cut problems on hyper-graphs \cite{chen2019suboptimality,gamarnik2021overlap}, as we will see our negative algorithmic results are closely linked to this property as well. 

\begin{figure}[t!]
    \centering
    \includegraphics[width=\textwidth]{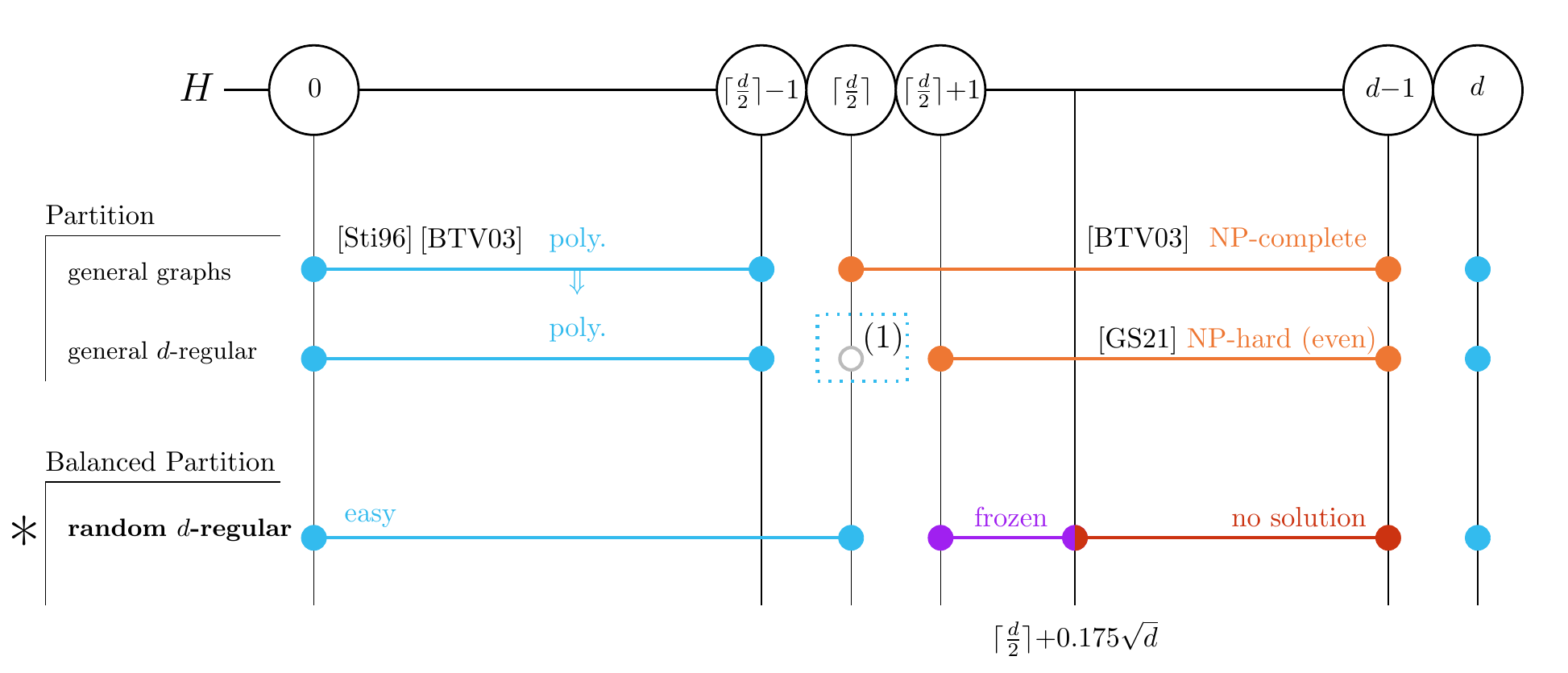}
    \caption{\textit{Overview of complexity results on assortative partitions.} 
    We compare the complexity of deciding whether an $H$-assortative partition exists for different ensembles of general and random graphs.
    For graphs with differing degrees per node, the degree $d$ should be interpreted as a function of each nodes degree $d(v)$.
    Therefore, $H$ also depends on a specific node.
    The ``?'' marks that to our knowledge there are no results on these sub-cases known in the literature.
    \textit{(1)} refers to a long-standing conjecture  that only finitely many regular graphs exist for which there are no $\lceil\frac{d}{2}\rceil$-assortative partitions.
    The * marks the results of this work.}
    \label{fig:related-work}
\end{figure}

\paragraph{Relation to local maximum cut, local minimum bisection and spin glasses}
The (maximum) minimum bisection is the bipartition of $G$ which (maximizes) minimizes the number of \textit{crossing} edges, i.e. edges with endpoints in both $V_-$ and $V_+$.
These classical problems are naturally related to the (anti-)ferromagnetic Ising model at zero magnetization,
since the ground state energy of its Hamiltonian exactly gives the size of the optimal bisection.

Locally maximal and minimal bisections were studied in~\cite{angelLocalMaxcutSmoothed2017,chenSmoothedComplexityLocal2019}.
In fact, any assortative bipartition for $H = \lceil d/2 \rceil$ is a locally minimal solution to the minimal bisection problem, as no node can be flipped to decrease the bisection size.
Analogously, a disassortative (bi)partition for $H = \lceil d/2 \rceil - 1$ is a locally maximal solution to the maximum (bisection)cut problem.
In fact, one can obtain a lower bound on the size of the maximum cut on sparse graphs using the assortative/disassortative properties of the global optimum  \cite{gamarnikMaxcutSparseRandom2018}.

(Dis)assortative partitions also correspond to single-spin flip stable state in the (anti-)ferro\-magnetic Ising model.
Using replica computations, the number of these states was computed for spin glasses on fully connected graphs by~\cite{brayMetastableStatesInternal1981}.
Their result is immediately meaningful for the solution space of the $\lceil d/2 \rceil$-assortative bipartitions and in this work we recover their results by means of a large $d$ expansion of the belief propagation equations as we will see in Section~\ref{sec:large-d-derivation}.

A striking relation between the ground state energies of the maximum and minimum bisection, i.e. the number of crossing edges of the globally optimal solutions, was discovered for  on random $d$-regular graphs.
For random $d$-regular graphs and large $n$, the cavity method from statistical physics predicts that the sizes of the minimum and maximum bisection are closely related with high probability \cite{zdeborovaConjectureMaximumCut2010}.
The same was predicted for the size of the maximum bisection and maximum cut, i.e. the ground state energy of the anti-ferromagnetic Ising model at arbitrary magnetization.
This conjecture was proven to be true asymptotically in $d$ in the first order and with high probability by \cite{demboExtremalCutsSparse2017}.
In this work, we make a similar connection between the entropy of the locally extremal solutions of both problems.

\paragraph{Hopfield Networks}
The existence of $\lceil \frac{d}{2} \rceil$-assortative partitions has sometimes been related to the stored patterns in a Hopfield network~\cite{bazganSatisfactoryGraphPartition2010}.
A Hopfield network is a model of associative memory \cite{hopfieldNeuralNetworksPhysical1982}.
Given an input pattern, among a collection of stored patterns, it retrieves the most similar one.
Technically, this network is a graph with $n$ nodes, where $n$ is the dimension of one such binary pattern.
By initializing the nodes with the input's values, one can retrieve the most similar pattern according to an asynchronous local update rule:
A node $i$ updates its own value to be in accordance with the majority of its neighbors values, weighted according to the edges connecting them to $i$.
As long as the number of patterns does not exceed the capacity of the network, one can ensure that stored patterns are both \textit{stable} under this local update rule and an \textit{attractors} of sufficiently associated inputs.
Therefore, a $\lceil \frac{d}{2} \rceil$-assortative partition resembles a stable state in a Hopfield network.

Note, however, that an absolutely key element of Hopfield networks is a very specific choice of the edge weights that are related to the patterns we aim to store, most commonly given by the Hebbian learning rule. In the assortative partition problem there is no such special choice of the interaction weights and thus one should not really think of the assortative partitions as the stored patterns in a Hopfield network. 

Another relation between assortative-partitions and Hopfield networks has been discussed in \cite{trevesMetastableStatesAsymmetrically1988a}.
The authors there computed the number of maximally stable states, i.e. states having the largest possible local magnetic field. Translated to our notation they concluded that stable states exist up to $H = \frac{d}{2} + 0.17566\sqrt{d}$, which exactly corresponds to our results at large $d$. They subsequently argue that the more stable the states are, the larger their basin of attraction and the more likely they should be discovered by local dynamics (such as the Hopfield flipping rule described above). As we will see this line of reasoning is flawed and the exact opposite is true, the larger the $H$ the harder it is computationally to find these more stable assortative partitions. This is another reason why the assortative partitions on a graph with random interactions are actually of limited interested in the context of the Hopfield networks.

\section{Main results}

We find several new results about both the assortative and disassortative partitions on large random $d$-regular graphs. We describe here those that we find the most remarkable and leave the full account of our results to the body of the paper.

\paragraph{Existence of assortative partitions}

Table \ref{tab:existence} summarizes our results on the existence of non-empty assortative \textit{bi}partitions in random $d$-regular graphs. The numbers given in the table are the largest $H$ for which these $H$-assortative balanced partitions exist with high probability. The middle (resp. left) column then gives the range of odd (resp. even) degrees $d$ for which the corresponding $H$ is the largest one. Moreover, we find that when non-empty assortative partitions exist most of them have $m=0$ in the limit $n\to \infty$.

Therefore we conjecture that the existence results for the bipartitions are valid for non-empty general partitions as well. 

We establish these results using non-rigorous tools from statistical physics. Their rigorous proofs would be an interesting topic for future work. We also conclude and conjecture that the entropy (logarithm of the number) of those partitions in the $n\to \infty$ limit (or regime) is given by the results of Belief Propagation. 
\begin{table}
    \centering
    \begin{tabular}{ l|c |c }
 \toprule
 largest achievable $H$ & odd degrees & even degrees  \\ 
 \midrule $H = \lceil d/2 \rceil$ & $3 \le d \le 29$ & $2 \le d \le 6$ \\  
 \hline $H = \lceil d/2 \rceil+1$ & $31 \le d \le 127$ & $8 \le d \le 70$  \\ 
 \hline $H = \lceil d/2 \rceil +2$ & $129 \le d \le 289$ & $72 \le d \le 200$  \\
 \hline $H = \lceil d/2 \rceil +3$ & $291 \le d \le 515$ & $202 \le d \le 394$  \\
 \hline $H = \lceil d/2 \rceil +4$ & $517 \le d \le 807$ & $396 \le d \le 654$  \\
 \hline $H = \lceil d/2 \rceil + 0.17566 \sqrt{d}$ &  $d \to \infty$ &  $d \to \infty$\\
 \bottomrule
\end{tabular}
\caption{Largest $H$-assortative bipartitions that exist with high probability in random $d$-regular graphs. All those with $H \ge \lceil d/2 \rceil+1$ are frozen 1RSB and hence conjectured to be algorithmically hard.}
    \label{tab:existence}
\end{table}

\paragraph{Algorithmic results} 

Our main algorithmic result is that for $H \le \lceil d/2 \rceil$ it is easy to find assortative partitions and bisections in random $d$-regular graphs. We indeed provide two algorithms for which we verify this numerically.

For $H > \lceil d/2 \rceil$ we conclude and conjecture that finding $H$-assortative partitions in random $d$-regular graphs is computationally hard. For this range of parameters the $H$-assortative partitions are organized in a way that is referred to in the literature as frozen-1RSB, implying that for a typical partition the closest one next to it (in Hamming distance) is an extensive distance far away. Other computational problems with frozen-1RSB, such as locked constraint satisfaction problems \cite{zdeborova2008locked,zdeborovaConstraintSatisfactionProblems2008}, binary perceptron \cite{krauthStorageCapacityMemory1989,aubinStorageCapacitySymmetric2019,perkinsFrozenRSBStructure2021}, or multi-index matching problem \cite{martinFrozenGlassPhase2004,martinRandomMultiindexMatching2005}, have been conjectured to be computationally hard. The hardness has been shown in a line of work on the overlap gap property for the class of local algorithms \cite{gamarnik2021overlap}. The structure of solutions in the frozen-1RSB phase is sketched in Figure~\ref{fig:phases}. 

Of course, it will not be easy to prove such a hardness result for a general algorithm. But clearly establishing the conjectured hardness for particular classes of algorithms is an interesting topic for future work. 

\begin{figure}[ht!]
    \centering
    \includegraphics[width=0.6\textwidth]{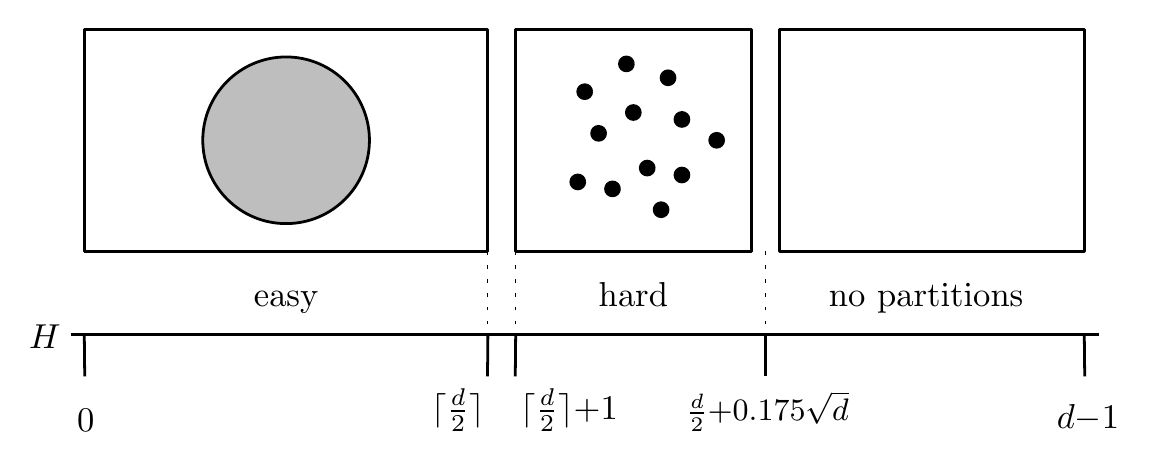}
    \caption{Phase diagram of the solution space for the assortative partition problem with the three different phases.}
    \label{fig:phases}
\end{figure}

\paragraph{Relation between the assortative and disassortative partitions problem, and spin glasses}

Extending the conjecture of \cite{zdeborovaConjectureMaximumCut2010} and for very analogous reasons we argue that on random regular graphs at $m=0$ there is a close relation between the assortative partition problems, the disassortative one, and a spin-glass version thereof. 

All the existence and algorithmic results stated above for the assortative case are hence valid also for the disassortative case with $H_{\rm dis} = d+1 - H_{\rm ass}$ and also for a spin glass with interactions $J_{ij} \in \{ \pm 1\}$ where the local magnetic field on every spin $i$, in its local neighborhood $\partial i = \{ j \in V | (j,i) \in E\} $, is required to be 
$$
     \sum_{j \in \partial i} J_{ij} x_i x_j \ge 2H_{\rm ass} - d \quad \forall i\, . 
$$
In the context of spin glasses these configurations for $H=\lceil d/2 \rceil$ are called the single spin flip stable states. For $H>\lceil d/2 \rceil$ these would correspond to configurations where every spin is stable with a margin/gap of $2H_{\rm ass} - d$. 

\paragraph{Conjectures based on comparisons with existing work}

Comparing our results to past works on (dis)assortative partitions or bisections and spin glasses we pose several questions/conjectures. We kindly ask readers to let us know if/when some of these are established. 

\begin{itemize}   
    \item {\bf Existence of $H$-partitions in general $d$-regular graphs:} Typical random $d$-regular graphs are structurally not so different from general $d$-regular graphs with girth bounded away from below.
    Given this, one may wonder whether our prediction of the existence of $H$-assortative partitions on random $d$-regular graphs is valid also on these general $d$-regular graphs. We conjecture this to be the case. 

   \item {\bf Positive algorithmic results for $H = \lceil d/2 \rceil$:} We give empirical evidence that indeed assortative partitions at $H = \lceil d/2 \rceil$ are found easily and so are disassortative ones. This result has been established at $d=n/2$ in \cite{ferberFriendlyBisectionsRandom2021a}. Future work should aim at a rigorous analysis of the related algorithms (or other ones) establishing that indeed these partitions can be found efficiently on random $d$-regular graphs for all $d$, or even for all general (perhaps large girth) $d$-regular graphs.
   
   \item {\bf Gapped states in spin glasses are computationally hard to find:} Our hardness conjecture has striking implications for the widely studied class of spin-glass models including the famous Sherrington-Kirkpatrick (SK) model \cite{sherrington1975solvable}. Single spin flip stable states, corresponding to $H = \lceil d/2 \rceil$, are algorithmically easy to find and widely studied, see e.g. \cite{brayMetastableStatesInternal1981}. Our results imply that ``gapped'' spin flip stable states, corresponding to $H > \lceil d/2 \rceil$ are computationally hard to find, yet we have not found any related investigation in the existing literature (not even an empirical one). 
   
   While we state most of our results for sparse random regular graphs, it is well known in the spin glass literature that the large $d$ expansion of the sparse case recovers the results on fully connected graphs. Let us thus define a gap $h = (2H - d)/\sqrt{d} = O(1)$ which then gives us for the fully connected limit
   $$
     \frac{1}{\sqrt{n}}\sum_{j \in \partial i} J_{ij} x_i x_j \ge h \quad \forall i\, , 
   $$
   where we recovered the scaling $1/\sqrt{n}$ well known from the SK model. 
   While we stated the equivalence between the assortative partitions and spin glasses with interactions $J_{ij} \in \{\pm 1\}$, it is well known that the distribution of interactions does not really influence macroscopic properties and hence this could imply that even in the canonical SK model states that have a gap bounded away from zero $h>\varepsilon$ for some $\varepsilon>0$ are computationally hard to find, while they exists for $h< 2\times 0.17566=0.35122$. Their existence up to this value of the gap $h$ was predicted in \cite{trevesMetastableStatesAsymmetrically1988a}.
   \item {\bf Omnipresent marginal stability in glasses is a consequence of computational hardness:}  We may speculate that the above computational hardness of finding gapped states in spin glasses is actually the reason behind the omnipresent observation that dynamics in glasses and spin glasses converge to marginally stable states \cite{muller2015marginal,parisi2017marginally}. This rather generic observation could very simply be a consequence of the fact that anything actually stable is computationally hard to access. The arguments about algorithmic hardness in this paper are restricted to random mean-field systems (e.g. sparse random graphs). For instance, if we could plant a gapped partition in a random graph that was easy to find, along the lines of \cite{zdeborova2011quiet}. The finite dimensional cases, where e.g. the corresponding graph is a lattice, may also behave differently. But in the view of the present results, investigating the question of hardness of finding gapped states in general systems including finite dimensional ones should be of interest. 
   \end{itemize}

\section{Assortative partitions for \texorpdfstring{$H> \lceil d/2 \rceil$}{H > ceil(d/2)} are frozen}\label{sec:locked-frozen}

We want to investigate how the solution space of the partition problems on $d$-regular graphs is arranged and consider the assortative partitions as an example.
If we have a valid $H$-assortative partition $P = \langle V_-, V_+ \rangle$ on a given graph, what do we need to do to obtain a new partition $P'$ that is also $H$-assortative, but where at least one node has changed its group compared to $P$?
If with a single flip of a node, we start a cascade that invalidates the assortativity of an extensive fraction of other nodes in the graph, the partitions $P$ and $P'$ are far away from one another in the solution space.

We now show that for some $H$ this is exactly the case and an extensive fraction of the nodes needs to change in order to arrive at the closest solution. The argument is adapted from \cite{zdeborova2011quiet}.
\begin{thm}\label{thm:extensive-fraction}
Let $d \geq 3$ be a fixed integer. Let $G(V,E)$ be any large $d$-regular random graph with $n$ nodes and assume that there exists a partition $P = \langle V_-,V_+\rangle$ for this graph which is an $H$-assortative partition with $H = \lceil d/2 \rceil + k$. 
\begin{itemize}
    \item For $d$ even, set $k\geq 2$.
    \item For $d$ odd, set $k\geq 1$.
\end{itemize}
Then, there exists an $\epsilon > 0$, such that no other solution $P'=\langle V'_-,V'_+\rangle$ exists which is at most $\epsilon n$ node flips from $P$. We call solutions satisfying this property \textit{frozen}. 
\end{thm}
To prove this fact, we can use that random regular graphs are expanders with high probability.
Let the \textit{edge expansion} or \textit{isoperimetric number } of a graph $G(V,E)$ be
$$
h(G,c) = \min_{S \subset V, |S| \leq c} \frac{|\partial S|}{|S|}\, ,
$$
where we define the edge boundary $\partial S = \{ (i,j) \in E | i \in S \land j \in \bar{S} \}$ and $\bar{S} = V \setminus S$.
\begin{thm}[\cite{hooryExpanderGraphsTheir2006}]\label{thm:expander}
Let $d\geq 3$ be a fixed integer. Then for every $\delta > 0$ there exists $\varepsilon > 0$ such that for almost every $d$-regular graph $G(V,E)$ with $n$ nodes
$$
h(G,\varepsilon n) \geq d-2-\delta\, .
$$
\end{thm}
\begin{proof}[Proof of Theorem~\ref{thm:extensive-fraction}.] 
We argue by contradiction and assume that there is indeed a  $H$-assortative partition $P'$ which is close to the given partition $P$.
Given some $\delta > 0$ we choose $\varepsilon > 0$ such that Theorem~\ref{thm:expander} holds.
We then assume that the set of flipped nodes between $P$ and $P'$ is $S = \{ j \in V | (j\in V_+ \land j\in V'_-) \vee (j \in V_- \land j \in V'_+)\}$ with a size of $|S| = \tilde{n}  \leq \varepsilon n$.
For a specific $H$ a subset of  the direct neighbors of $i \in S$ must also be contained in $S$.
In particular, for $k \geq 0$ and $i \in V_+ \cap S$,  at least $2k$ neighbors from $\partial i \cap V_+$ need to change to $V_-'$ together with $i$ for even $d$ and likewise for $V_-$.
This implies the existence of edges within $S$.
Since $G$ is $d$-regular, we can bound the number of edges that emanate from $S$ to $\bar{S}$ as $|\partial S| \leq (d-2k)\tilde{n}$ and therefore $h(G,\tilde{n}) \leq d-2k$.
By Theorem~\ref{thm:expander} we get that $d-2k \geq h(G,\tilde{n}) \geq d-2-\delta$.
Since we can select $\delta$ as small as we want, this inequality is only valid when $k < 2$.
For larger $k$, it does not hold and therefore our assumption, that the solutions $P$ and $P'$ are $\tilde{n}\leq \varepsilon n$ close is wrong which implies that Theorem~\ref{thm:extensive-fraction} holds.
We can repeat the same arguments for odd $d$ where we need to flip at least $2k+1$ neighbors of a node $i \in S$ and arrive at a contradiction for $k \geq 1$.
\end{proof}

However, our proof of Theorem~\ref{thm:extensive-fraction} gives only a lower bound on the distance between solutions. 
It could very well be necessary to change even more nodes than implied there to arrive at the closest other solution.
This is in particular important for $k=1$ and even degree $d$ which are not included in the Theorem, since the flip of a node implies that least two of its neighbors flip as well.
If these flipped nodes are arranged in a cycle, one would have a solution that is at $O(\log(n))$ distance to the given $P$ and therefore not extensive.
But flipping only these loops might not be sufficient typically. 
This is because flipping two of $i$'s neighbors is also only a lower bound, and it might be necessary to change $3$ or more, which would propagate the changes further on the graph and eventually change an extensive fraction nonetheless. 
In order to resolve how many nodes typically need to be flipped in random regular graphs, in particular for $k=1$, we develop the survey propagation in Section~\ref{sec:survey-popagation-theory} that tracks such required changes.
Using survey propagation, we find that actually the problem is frozen, i.e. valid configurations are at extensive distance from each other, for all $H > \lceil d / 2 \rceil$.
An analogous set of arguments can be applied to the disassortative partitions, where the freezing occurs at $H \leq \lceil d/2 \rceil$.

Theorem \ref{thm:extensive-fraction} implies that under the same assumptions the problem has the overlap gap property which is a sufficient condition for hardness for a number of algorithm classes, including local algorithms, AMP, low-degree polynomials and Boolean circuits \cite{gamarnik2021overlap}. 
What remains to be shown is that the partitions satisfying the assumptions of Theorem \ref{thm:extensive-fraction} exist in random $d$-regular graphs. This is addressed in the next section. 

\section{Replica Symmetric solution via Belief Propagation}\label{sec:rs}

\subsection{Probabilistic setting of the problem} \label{sec:rs-1}
\paragraph{Probability distribution}
Any partition $\langle V_-,V_+ \rangle$ on the graph can be encoded into a vector $\mathbf{x} = (x_1,x_2,...,x_n) \in \{\pm 1\}^n$.
Then, we can define a probability distribution on the space of possible valid partitions as follows
\begin{align}
\text{Prob}\Big(\mathbf{x}\Big) 
&= \frac{1}{Z} \exp\left(-\mu \sum_{i\in V}x_i \right) \prod_{i \in V}  \mathcal{C}_H\left(x_i\sum_{j \in \partial i}J_{ij} x_j\right)\, , \label{eq:prob-dist}
\end{align}
where $Z$ is the normalizing constant and $\mu \in \mathbb{R}$ is the chemical potential used to adjust the magnetization $m$ to a desired value. 
The $\mathcal{C}_H$ expresses a constraint on the exclusive neighborhood $\partial i$ of each node $i$
\begin{align}
\mathcal{C}_H\left(z\right) &= \ind\left(H \leq \frac{d+z}{2}\right)\, .\label{eq:constraint}
\end{align}
Here, $\ind(B)$ is the indicator function of a boolean statement $B$; if $B$ is true it takes value $1$ and $0$ otherwise.
If we set $J_{ij} = +1$ for all $(ij) \in E$ we obtain the assortative partition problem. Indeed, the sum $x_i\sum_{j \in \partial i} x_j= \sum_{j \in \partial i}\ind(x_i = x_j) - \sum_{j \in \partial i} \ind(x_i \neq x_j)$ counts how many more neighbors of $i$ are in the same group as the central node $i$ over those that are in the other group.
In the constraint $C_H$ this is converted to exactly the number of nodes in the neighborhood that are in the same group as $i$.
Thus, the constraint is fulfilled so long as $i$ has at least $H$ neighbors in its own group.
This means that we have an $H$-assortative partition if and only if all constraints $C_H$ are satisfied.
Consequently,  for $J_{ij}=+1$ the probability distribution \eqref{eq:prob-dist} only has non-zero probability on $H$-assortative partitions.
Additionally, the factor $\exp\left(-\mu \sum_{i\in V}x_i \right)$ weights partitions according to the size of the two groups, using the parameter $\mu$ that is referred to as the chemical potential in physics. The chemical potential plays the role of a Lagrange parameter, and adjusting its values suitably allows to sample only solutions of a given magnetization $m$.

We can proceed similarly to describe the disassortative partitions by setting $J_{ij} = -1$ for all $(ij)\in E$. The sum $-x_i\sum_{j \in \partial i} x_j$ then counts how many neighbors of $i$ are in the opposite group and the probability distribution (\ref{eq:prob-dist}) is therefore non-zero only on $(d-H+1)$-disassortative partitions.

\paragraph{Obtaining the number of valid partitions}
Given this definition of the probability distribution, we are able to compute the number of partitions that have non-zero probability for a given magnetization $m$.
For given values of $J_{ij}$, $H$ and $m$ we define this number to be $\mathcal{N}(m)$ and its entropy $s(m)$ (here we do not make dependence on $J_{ij}$ and $H$ explicit)
\begin{align}
    \mathcal{N}(m) = e^{ns(m)}\, ,
\end{align}
where $s(m)$ is the entropy.
Defining the free entropy density to be $\Phi(\mu) = \log(Z)/n$, we can split the partition function $Z$ according to the magnetization $m$ as follows
\begin{align}
    e^{n\Phi(\mu)} = Z = \int_{\{\mathbf{x} | \forall i \in V: \mathcal{C}_H(x_i \sum_{j \in \partial i} J_{ij}x_j) = 1 \}} e^{-\mu \sum_{i \in V} x_i} = \int_{m\in [-1,1]} {\rm d}m \hspace{0.5em}e^{ns(m)} e^{-nm\mu}\, .
\end{align}
For large graphs where the number of nodes goes to infinity we can apply the saddle point method to see that 
\begin{align}
    \Phi(\mu) &= s(\hat{m}) - \mu \hat{m} \, ,\\
    \frac{\partial s(m)}{\partial m}\Big|_{m=\hat{m}} &= \mu\, .
\end{align}
In order to obtain the function $s(m)$, we need to compute the magnetization $m(\mu)$ and the free entropy density $\Phi(\mu)$ for different values of $\mu$.
Both of them can be computed from the partition function $Z$, as the free entropy density is simply its logarithm and we can recover the magnetization as 
\begin{align}
    \frac{\partial \Phi(\mu)}{\partial \mu} =  \mathbb{E}_{\mathbf{x}}[m]  \, .
\end{align}
Generally, the high-dimensional sum $Z$ is intractable, but on tree graphical models the established tool of belief propagation (BP) \cite{yedidia2003understanding} allows to compute $Z$ exactly.
On large locally tree-like graphical models, analysis based on BP may be asymptotically exact under conditions discussed subsequently.

\begin{figure}
    \centering
    \includegraphics[width=0.3\textwidth,align=c]{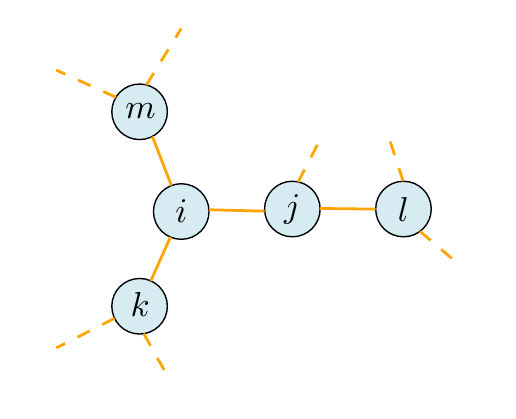}
    \includegraphics[width=0.65\textwidth,align=c]{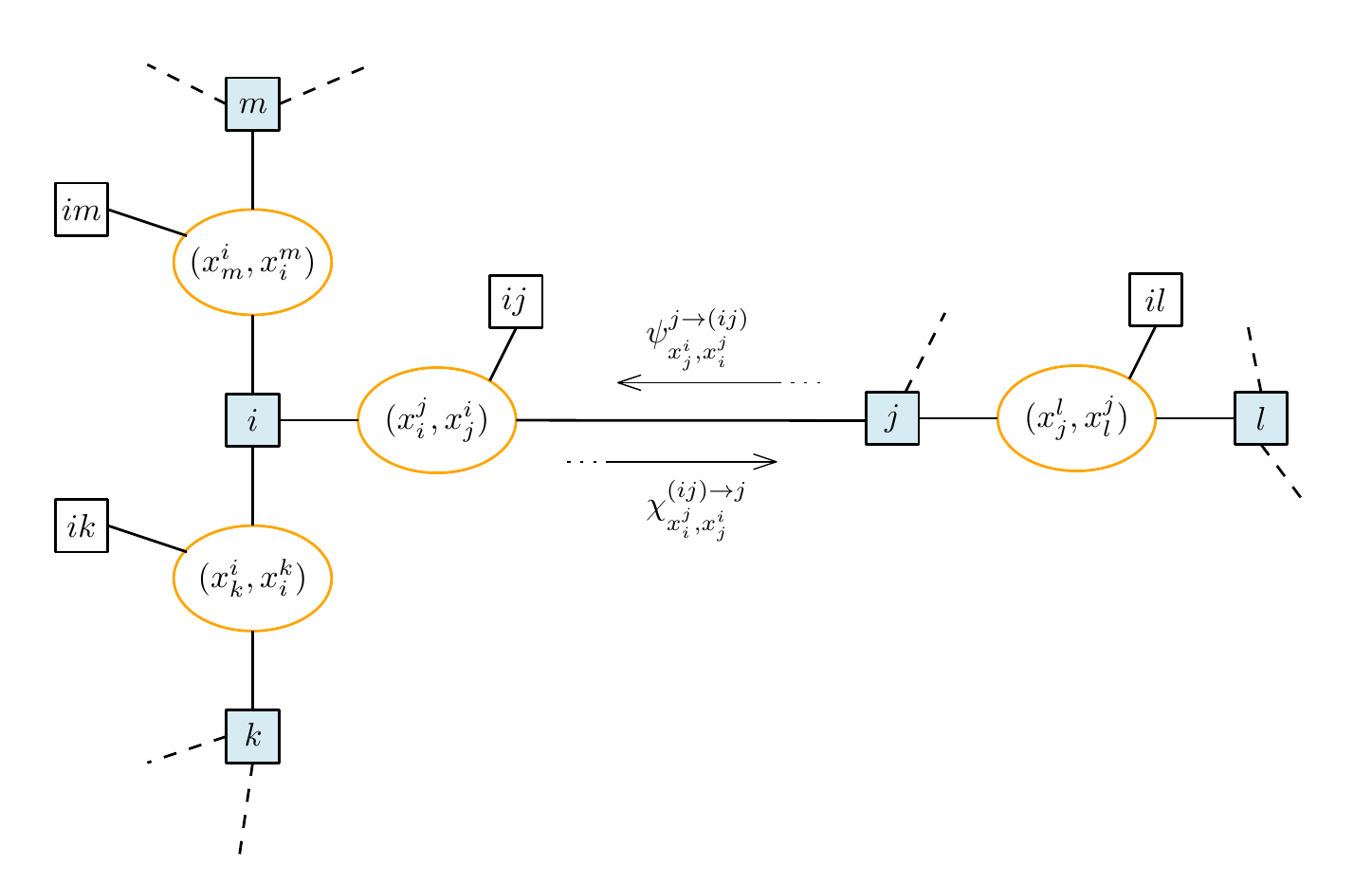}
    
    \caption{A subgraph of the factor graph for the assortative partition problem on a $3$-regular graph. \textit{(Left)} Original graph. \textit{(Right)} Factor graph. 
    The tuples in the variable nodes \vcenteredinclude{\includegraphics[width=0.03\textwidth]{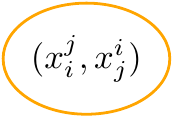}} can take on values in the possible groups in the partition. 
    The factor~nodes~\vcenteredinclude{\includegraphics[width=0.02\textwidth]{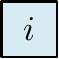}} enforce the constraint on assortativeness via $f_i$.
    The factor~nodes~\vcenteredinclude{\includegraphics[width=0.02\textwidth]{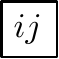}} correspond to the magnetization constraint.
    Messages $\psi$ and $\chi$ are sent along the edges.}
    \label{fig:factor-graph}
\end{figure}

\subsection{Belief Propagation on a tree-like factor graph}

\paragraph{Tree-like factor graph}
To use the Belief Propagation and compute the corresponding entropy, we need to find the factor graph associated to the probability measure \eqref{eq:prob-dist}.
In the most straightforward choice, every node in the original graph becomes a variable and each neighborhood of $i$ becomes a factor to ensure the constraint~$\mathcal{C}_H$.
Priors on the variables account for the chemical potential $\mu$. Unfortunately, this factor graph of \eqref{eq:prob-dist} is not tree-like by default, since the constraints on the assortativeness of node $i$ depend on all its neighbors $\partial i$.
Since their own assortativeness reciprocally also depends on $i$, this relation is responsible for creating small loops in the factor graph.
This poses a problem for applying belief propagation which is exact on trees and under specific conditions can be asymptotically exact on  tree-like graphs.

To eliminate the problematic small loops, we construct a different factor-graph by introducing auxiliary variables.
For each node membership $x_i$ we create exactly $|\partial i|$ copies (replicas) of it.
We call each replica $x_i^j$, where $j \in \partial i$ denotes the neighbor of $i$ to which this replica is associated.
We then obtain the desired tree-like factor graph on tuple variables, as depicted in Figure~\ref{fig:factor-graph}.
The prior contributions enforcing the chemical potential can then be expressed as
\begin{align}
              g_{ij}(x_i^j, x_j^i) &= \exp\left\{-\mu\left( \frac{x_i^j}{|\partial i|} + \frac{x_j^i}{|\partial j|} \right)\right\}\end{align}
and the factor nodes representing the neighborhood of node $i$ are
\begin{align} f_{i}(\{x_i^k\}_{k \in \partial i \setminus j},\{x_k^i\}_{k \in \partial i \setminus j}) &= \mathcal{C}_H\left( \sum_{j \in \partial i} J_{ij} x_i^j x_j^i \right)\ind\left(\left|\sum_{j\in \partial i} x_i^j\right|=|\partial i|\right)\, . 
        \end{align}
The second factor in $f$ is added to enforce the replicas of one node to all be the same so that the resulting probability distribution of the tree-like factor graph exactly corresponds to the probability distribution in \eqref{eq:prob-dist}.
In contrast to the loopy factor graph directly obtained from \eqref{eq:prob-dist} the variables on the new factor graph are tuples of node memberships $(x_i^j,x_j^i)$ instead of single memberships.

\paragraph{Belief propagation}
We then readily apply the Belief Propagation equations \cite{yedidia2003understanding} that will provide an asymptotically exact solution under the assumption that the incoming messages are independent. We will discuss how to self-consistently check this assumption in the next section. In statistical physics the analysis based on Belief Propagation is called the replica symmetric analysis, as the assumptions under which Belief Propagation gives the asymptotically exact result are equivalent to those of replica symmetry \cite{mezard1987spin,mezard2009information}.

We call $\psi_{x_i^j,x_j^i}^{i \to (ij)}$ the message sent by the factor node $(i)$ to the variable node $(ij)$, and $\chi_{x_i^j,x_j^i}^{(ij)\to j}$ the message sent by the variable node $(ij)$ to the factor $(j)$.
The oriented tuples $(x_i^j,x_j^i)$ contain the replicas of the groups $i$ and $j$ are assigned to.
According to belief propagation, they can be computed as
\begin{align}
     \psi_{x_i^j,x_j^i}^{i \rightarrow (ij)} &= \frac{1}{Z^{i \rightarrow (ij)}} \sum_{\{x_i^k,x_k^i\}_{k\in\partial i \setminus {j}}} \mathcal{C}_H\left( \sum_{j \in \partial i} J_{ij} x_i^j x_j^i \right)\ind\left(\left|\sum_{j\in \partial i} x_i^j\right|=|\partial i|\right) \prod_{k\in \partial i \setminus j} \chi _{x_k^i,x_i^k}^{(ki)\rightarrow i}\, ,\label{eq:hugefixpoint} \\
    \chi_{x_i^j,x_j^i}^{(ij)\rightarrow j} &= \frac{1}{Z^{(ij)\rightarrow j}} \exp\left\{-\mu\left( \frac{x_i^j}{|\partial i|} + \frac{x_j^i}{|\partial j|} \right)\right\}\psi_{x_i^j,x_j^i}^{i\rightarrow(ij)}\, ,
\end{align}
where the $Z^{i \rightarrow (ij)} , Z^{(ij)\rightarrow j}$ are normalisation constants to ensure that $\psi^{i \rightarrow (ij)}, \chi^{(ij)\rightarrow j}$ are probability distributions.

The indicator function in the factor to variable messages $\psi_{x_i^j,x_j^i}^{i \rightarrow (ij)}$ is zero when the replicas of $i$ are not all identical.
From this, we can safely assume that at the fixed point for all $i$ it holds that $x_i = x_i^j$ for $j \in \partial i$.
This simplifies the messages greatly and we obtain
\begin{align}
     \psi_{x_i,x_j}^{i \rightarrow (ij)} &= \frac{1}{Z^{i \rightarrow (ij)}} \sum_{\{x_k\}_{k\in\partial i \setminus {j}}} \mathcal{C}_H\left( \sum_{j \in \partial i} J_{ij} x_i x_j \right) \prod_{k\in \partial i \setminus j} \chi _{x_k,x_i}^{(ki)\rightarrow i}\, ,\label{eq:333} \\
    \chi_{x_i,x_j}^{(ij)\rightarrow j} &= \frac{1}{Z^{(ij)\rightarrow j}} \exp\left\{-\mu\left( \frac{x_i}{|\partial i|} + \frac{x_j}{|\partial j|} \right)\right\}\psi_{x_i,x_j}^{i\rightarrow(ij)}\, .
\end{align}
For a fixed point of these messages, the corresponding replica symmetric free  entropy $\Phi_{\text{RS}}$ can be formulated as
\begin{align}
    n\Phi_{\text{RS}} &= 
    \sum_{(ij)\in E}\log \left(Z^{(ij)}  \right) 
    + \sum_{i \in V}\log \left(Z^{i} \right) 
    -  2\sum_{(ij)\in E}\log \left(Z^{(ij),i} \right)
    \, ,\\
    Z^{(ij)} &= \sum_{\{x_i,x_j\}} \exp\left\{-\mu\left( \frac{x_i}{|\partial i|} + \frac{x_j}{|\partial j|} \right)\right\}\psi_{x_i,x_j}^{i\rightarrow (ij)} \psi_{x_j,x_i}^{j\rightarrow (ji)}\, ,\\
     Z^{i} &= \sum_{x_i,\{x_{k}\}_{k\in\partial i}}\mathcal{C}_H\left( \sum_{j \in \partial i} J_{ij} x_i x_j \right)\prod_{j \in \partial i} \chi_{x_j,x_i}^{(ji) \rightarrow i}\, ,\\
    Z^{(ij), i} &= \sum_{\{x_i,x_j\}} \chi_{x_j,x_i}^{(ji)\rightarrow i}\psi_{x_i,x_j}^{i \rightarrow (ij)}\, .
\end{align}
Substituting $\psi_{x_i,x_j}$ in this expression by rearranging \eqref{eq:333}, some parts simplify and we get
\begin{align}\label{eq:entropyyy}
     n\Phi_{\text{RS}} &=
    \sum_{i \in V}\log \left(Z^{i} \right) - \sum_{(ij)\in E}\log \left(Z^{ij} \right)\, ,\\
     Z^{i} &= \sum_{x_i,\{x_{k}\}_{k\in\partial i}}\mathcal{C}_H\left( \sum_{j \in \partial i} J_{ij} x_i x_j \right)\prod_{j \in \partial i} \chi_{x_j,x_i}^{(ji) \rightarrow i}\, ,\label{eq:partition-sum}\\
     Z^{ij} &= \sum_{\{x_i,x_j\}}
     \exp\left\{\mu\left( \frac{x_i}{|\partial i|} + \frac{x_j}{|\partial j|} \right)\right\}
    \chi_{x_i,x_j}^{(ij)\rightarrow j}
    \chi_{x_j,x_i}^{(ji)\rightarrow i}\, . \label{eq:entropyyyend}
\end{align}
Likewise we can eliminate $\psi_{x_i,x_j}^{i \rightarrow (ij)}$ in the update rule and we combine the computation of two message types into only a single message 
\begin{align}
       \chi_{x_i,x_j}^{i \to j} &= \chi_{x_i,x_j}^{(ij)\rightarrow j}= \frac{1}{Z^{i \to j}} \exp\left\{-\mu\left( \frac{x_i^j}{|\partial i|} + \frac{x_j^i}{|\partial j|} \right)\right\} \sum_{\{x_k\}_{k\in\partial i \setminus {j}}} \mathcal{C}_H\left( \sum_{j \in \partial i} J_{ij} x_i x_j \right) \prod_{k\in \partial i \setminus j} \chi_{x_k, x_i}^{k \to i}\, ,\label{eq:final-fixed-point}
\end{align}
where we redefine $Z^{i\to j} = Z^{i \rightarrow (ij)}Z^{(ij) \rightarrow j}$.
From this we can compute the magnetization $m$ as a function of $\mu$ by 
\begin{align}
     m &= \frac{\partial \Phi_{RS}(\mu)}{\partial \mu} 
     %= -\frac{d}{2} \sum_{(ij) \in E} \frac{1}{Z^{ij}}\frac{\partial  Z^{ij}}{\partial \mu}
     = -\frac{d}{2}\sum_{(ij) \in E} \frac{ \sum_{(x_i,x_j) \in \{\pm 1\}^2} \left( \frac{x_i}{|\partial i|} + \frac{x_j}{|\partial j|} \right)  e^{\mu \left( \frac{x_i}{|\partial i|} + \frac{x_j}{|\partial j|} \right)}  \chi_{x_i,x_j} \chi_{x_j,x_i}}{\sum_{(x_i,x_j) \in \{\pm 1\}^2}   e^{\mu \left( \frac{x_i}{|\partial i|} + \frac{x_j}{|\partial j|} \right)}  \chi_{x_i,x_j} \chi_{x_j,x_i}}\, .
\end{align}
Now we have all the ingredients ready to compute the entropy $s(m)$ on a full graph as
\begin{align}\label{eq:entropy-computation}
    s(m) = \Phi_{\rm RS}(\mu) + \mu m.
\end{align}
We only need to find the fixed point $\chi$ for some given value of $m$ and calculate $\Phi_{\rm RS}$ at this fixed point from \eqref{eq:entropyyy}--\eqref{eq:entropyyyend}.
Inserting this into \eqref{eq:entropy-computation}  enables us to find the entropy at the corresponding $\mu$.

Algorithm \ref{alg:full-BP} shows the procedure employed to find a fixed point of the BP.
Note that an update on $\mu$ needs to be included to enforce the fixed point to have magnetization $m=m_{\rm target}$ rather than $m = \pm 1$.
At convergence, $\Phi_{\rm RS}$ can be computed.

\begin{algorithm}
\caption{Belief Propagation on a full graph $G(V,E)$}\label{alg:full-BP}
\textbf{In:} 
$G(V,E); \chi = \{\chi^{i \to j}\}_{(ij) \in E}$; $m_{\rm target} \in [-1,1]$\;
\While{$\chi$ not converged}
{
$\mu \gets - \argmin_{\mu \in \mathbb{R}}| m(\mu, \chi) - m_{\rm target}|$\;
select uniformly random $i \in V$\;
select uniformly random $j \in \partial i$\;
\For{$(x_i,x_j) \in \{\pm 1\}^2$}{

$\chi^{i \to j}_{x_i,x_j} \gets \exp\left\{-\mu\left( \frac{x_i^j}{|\partial i|} + \frac{x_j^i}{|\partial j|} \right)\right\} \sum_{\{x_k\}_{k\in\partial i \setminus {j}}} \mathcal{C}_H\left( \sum_{j \in \partial i} J_{ij} x_i x_j \right) \prod_{k\in \partial i \setminus j} \chi_{x_k, x_i}^{k \to i}$\;
}
$\chi^{i \to j} \gets \frac{\chi^{i \to j}}{|\chi^{i \to j}|}$\;
}
return $\chi$, $\mu$\;
\end{algorithm}
\newpage

\subsection{BP on \texorpdfstring{$d$}{d}-regular graphs}\label{sec:BP-dreg}
We now consider the problem only on $d$-regular random graphs and $J = J_{ij}$ for all $i,j$.
Thus, also $|\partial i| = |\partial j| = d$ for all $i,j$.
Under the same assumptions under which the Belief Propagation is asymptotically exact on tree-like graphs, this implies, that at a fixed point of the belief propagation, all messages have to be equal locally and we set them to be
\begin{align}
    \chi_{x,y} = \chi_{x,y}^{i\rightarrow j}, \quad \forall i, j \in V\, .
\end{align}
Note, that the tuples of membership assignments $(x,y)$ are still oriented, which needs to be taken into account when formulating the update rules.
We can use this to simplify the BP update rule considerably by applying this assumption to \eqref{eq:final-fixed-point} to get
\begin{align}
    \chi_{x,y} &= \frac{1}{Z} e^{-\frac{\mu}{d}\left( x + y \right)}\sum_{r=0}^{d-1} \binom{d-1}{r} \mathcal{C}_H\left(J \,  \frac{2r}{d}\right)(\chi_{x,x})^r (\chi_{-x,x})^{(d-1-r)}\, .
\end{align}
By introducing the sum formulation we only count the assignments on the remaining $d-1$ edges that fulfill the constraint $\mathcal{C}_H$. 
The binomial coefficient is present to account for the number of combinations in which this can be achieved. 
The replica symmetric free entropy on $d$-regular graphs reduces to
\begin{align}\label{eq:Entropy-RS}
    \Phi_{\text{RS}} &=
    \log \left(Z_{\rm node} \right) - \frac{d}{2} \cdot \log \left(Z_{\rm edge} \right)\, ,\\
    Z_{\rm node} &= \sum_{x \in \{\pm 1\}} \sum_{r=0}^d \binom{d}{ r} \mathcal{C}_H\left(J \,  \frac{2r}{d}\right)(\chi_{x,x})^r (\chi_{-x,x})^{(d-r)}\, ,\\
    Z_{\rm edge} &= \sum_{(x,y) \in \{\pm 1\}^2} e^{\frac{\mu}{d}\left( x + y \right)} \chi_{x,y} \chi_{y,x}\, .
\end{align}
The magnetization simplifies to 
\begin{align}
     m &= -\frac{1}{2} \frac{ \sum_{(x,y) \in \{\pm 1\}^2}  (x+y)  e^{\frac{\mu}{d}\left( x + y \right)}  \chi_{x,y} \chi_{y,x}}{\sum_{(x,y) \in \{\pm 1\}^2}  e^{\frac{\mu}{d}\left( x + y \right)}  \chi_{x,y} \chi_{y,x}}\, .
\end{align}
For the case of $H_{\rm ass}$-assortative partitions with $J=+1, H= H_{\rm ass}$ the BP update and computation of $Z_{\rm node}$ simplify into
\begin{align}
    \chi_{x,y}^{\rm ass} &= \frac{1}{Z}  e^{-\frac{\mu}{d}\left( x + y \right)} \sum_{r = \minH_{\rm ass} - \ind(x = y)}^{d-1} \binom{d-1}{ r} (\chi_{x,x})^r (\chi_{-x,x})^{(d-1-r)}\, , \label{eq:dreg-fixed-point}\\
    Z_{\rm node}^{\rm ass} &= \sum_{x \in \{\pm 1\}} \sum_{r=\minH_{\rm ass}}^d \binom{d}{r} (\chi_{x,x})^r (\chi_{-x,x})^{(d-r)}\, .
\end{align}
The last line stems from the fact that the factor for a specific node in partition $x$ with a given neighbor in partition $y$ and a remaining neighborhood of size $d-1$ is only non-zero if it has at least $\minH - \ind(x = y)$ neighbors that are in the same partition as $x$.
Only if this is fulfilled, the central node is assortative.

Applying the same for $H_{\rm dis}$-disassortative partitions with $J=-1, H= d-H_{\rm dis}+1$ we have
\begin{align}
    \chi_{x,y}^{\rm dis} &= \frac{1}{Z}  e^{-\frac{\mu}{d}\left( x + y \right)} \sum_{r = 0}^{\minH_{\rm dis}  - \ind(x = y) - 1} \binom{d-1}{ r} (\chi_{x,x})^r (\chi_{-x,x})^{(d-1-r)}\, ,\\
    Z_{\rm node}^{\rm dis} &= \sum_{x \in \{\pm 1\}} \sum_{r=0}^{\minH_{\rm dis}-1} \binom{d}{ r} (\chi_{x,x})^r (\chi_{-x,x})^{(d-r)}\, .
\end{align}

For these formulas, we can use the same procedure as before for the BP on the full graph to find the entropy at a specific magnetization for $d$-regular graphs.
The only difference will be, that we iterate on exactly one message instead of on two per every edge and instead of updating every edge randomly we introduce dampening.

\paragraph{Equality between assortative and disassortative bipartitions}
From the fixed point equations of the (dis)assortative partitions, we can immediately see that for $\mu=0$ there is a strong formal resemblance between the two problems.
This is the case when there is no external fields and the fixed points converge to a solution for $m=0$, a bipartition.
For assortative bipartitions with a fixed point $\hat{\chi}^{\rm ass}$ and $H_{\rm ass}$, it follows that $\hat{\chi}^{\rm dis}$ is a fixed point for the disassortative case with $H_{\rm dis} = d - H_{\rm ass} + 1$:
\begin{align}
    \hat{\chi}^{\rm dis}_{x,x} &= \hat{\chi}^{\rm ass}_{-x,x} \, ,\label{eq:equality1}\\
    \hat{\chi}^{\rm dis}_{-x,x} &= \hat{\chi}^{\rm ass}_{x,x}\, , \forall x \in \{\pm 1\}\, . \label{eq:equality2}
\end{align}
Moreover, under these conditions $\Phi_{\rm RS}^{\rm ass} = \Phi_{\rm RS}^{\rm dis} $ is equal, as $Z_{\rm node}^{\rm dis} = Z_{\rm node}^{\rm ass}$ and $Z_{\rm edge}$ does not change under \eqref{eq:equality1}--\eqref{eq:equality2}.
Note, that this equality does not hold for $\mu \neq 0 $, since for imbalanced partitions the elements of $\chi$ are weighted differently and we have different fixed point equations.

In sight of this equality result, we will analyze the BP-based results for both the assortative and disassortative partition problem interchangeably in the case of zero magnetization. We also note that for  regular random graphs with $m=$ the replica symmetric entropy is generically equal to the annealed entropy \cite{mora2007geometry} and is thus an upper bound on the true entropy.

\subsection{Analysis of BP results for \texorpdfstring{$d$}{d}-regular random graphs}\label{sec:bp-results}
\begin{figure}[b]
     \centering

     \begin{subfigure}[b]{0.49\textwidth}
         \centering
         \includegraphics[width=\textwidth]{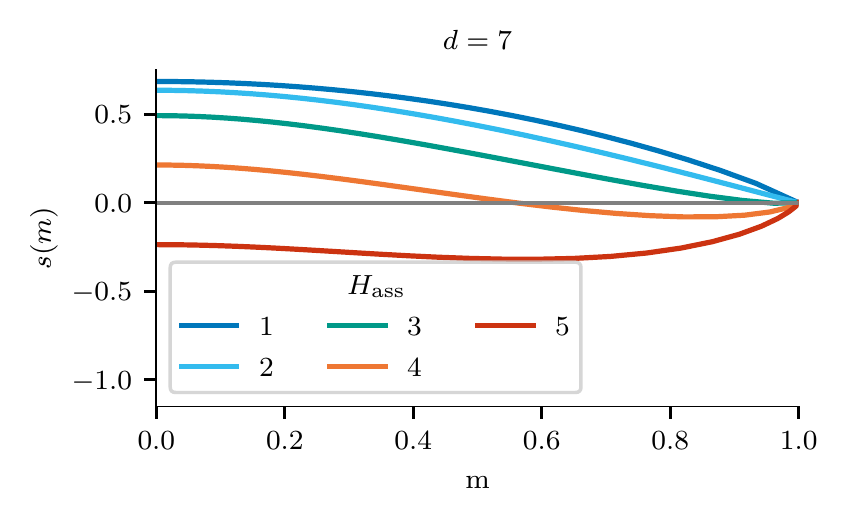}
          \caption{assortative}
        \label{fig:vary-m-ass}
     \end{subfigure}
     \centering
     \begin{subfigure}[b]{0.49\textwidth}
         \centering
         \includegraphics[width=\textwidth]{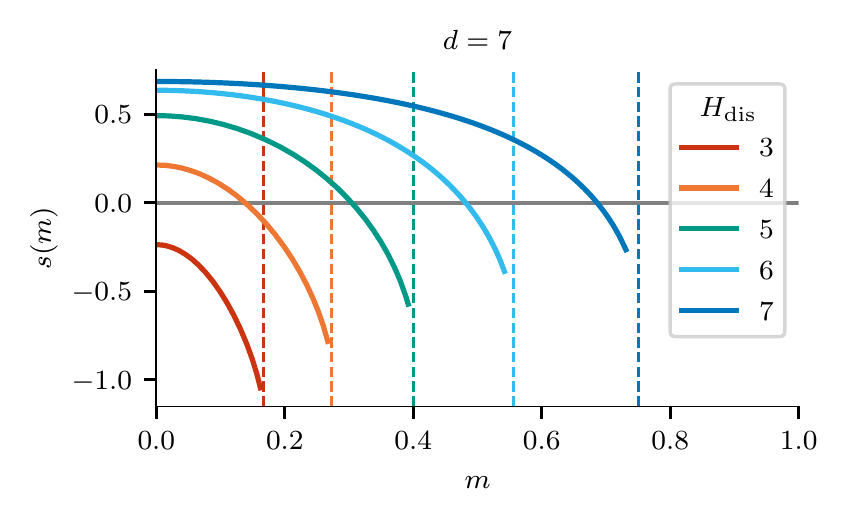}
           \caption{disassortative}
        \label{fig:vary-m-dis}
     \end{subfigure}
     \caption{\textit{Entropy $s(m)$ vs. magnetization $m$ on $7$-regular graphs.} Left panel is for assortative partitions, right panel for disassortative ones. 
     The colors are mapped such that same colored lines have the same $s(0)$.
     For the disassortative partition, dashed lines indicate the magnetization $\tilde{m}(H_{\rm dis})$, at which the exact number of minimally required edges of the bigger group to the smaller one is larger than the the number of edges the smaller group can maximally be connected to.
     } \label{fig:vary-m}
\end{figure}
%\paragraph{Entropy as a function of magnetization $m$}
In order to obtain the replica symmetric entropy, the fixed point equation~\eqref{eq:dreg-fixed-point} is iterated until convergence, adapting $\mu$ at every iteration in order to obtain the desired value of $m$. 
Because the problems are symmetric over $\pm m$, we only show and discuss results for positive $m$.

The resulting entropy as a function of $m$ is depicted in Figure~\ref{fig:vary-m} both for the assortative (left) and disassortative (right) case, for degree $d=7$.
The general behavior is consistent for various values of $d$.

We exactly recover the predicted equality at $m=0$ between the entropy of the assortative and of the  disassortative case, when the $H$ is reparameterized properly, $H_{\rm ass} = d+1-H_{\rm dis}$.
This correspondence is highlighted by the colors of the corresponding curves and shown separately in Figure~\ref{fig:comp-happy-unhappy}.
In both cases this is also the maximum of the entropy over all $m$, which means that if partitions exist most of them are (close-to) bipartitions. 
%\todo{Discuss implication for finding bipartitions vs partitions.}
Beyond $m=0$, the entropy of the two assortative and disassortative problems behave very differently.
For the assortative partitions, at $m=1$ the entropy is always exactly zero, as there is only exactly one satisfying partition: The one with all nodes in one group, the partition $\langle V, \varnothing \rangle$. 
%There is no similar argument that holds for the assortative case.
The two cases $H_{\rm ass} = 1,2$ are a bit special as for them the entropy never goes below zero.
For $H_{\rm ass}=1$ it is enough for a node to have one of their neighbors in the partition. 
So as long as the number of neighbors in each partition is even, one can easily place these pairs everywhere in the graph and find such partitions.
For $H_{\rm ass}=2$ the picture is similar, as one can place loops of size $\log(n)$, which are present in large $d$-regular graphs and which make the nodes in that partition $2$-assortative.
The problem looks different though for $H_{\rm ass} \ge 3$, because the both groups must contain at least a $3$-regular graph as a subgraph, as we discussed in Section~\ref{sec:locked-frozen}. All such subgraphs are with high probability extensively large, because they are expanders. Thus, there cannot exist partitions with magnetization close to $1$. We indeed see in Figure~\ref{fig:vary-m-ass} that for $H_{\rm ass} \ge 3$ in an interval around $m=1$ the entropy is negative. Under the assumptions of BP, the magnetizations for which the entropy remains positive are those for which $(H,m)$-partitions exist with high probability. The entropy at $m=0$ turning negative signals the absence of all non-empty $H$-partitions. 
%For this case, the entropy actually goes below $0$ for larger $m$ in .

For the disassortative case the entropy curves look a bit different at $m\ne 0$. They become negative at some value of $m$ for all $H_{\rm dis}$. They actually diverge to $-\infty$ before they reach $m = 1$. In fact, they diverge exactly at $\tilde{m}(H_{\rm dis})$, which is the imbalance between the partitions at which a valid disassortative partition becomes combinatorially impossible, i.e. at
\begin{align}
   \tilde{m}(H_{\rm dis}) = \min\Big\{m \in (0,1) \Big| \frac{m+1}{2} \cdot H_{\rm dis} <  (1-\frac{m+1}{2}) \cdot d \Big\} \, .
\end{align}
At this value, the number of crossing edges required to satisfy the nodes in the larger partition is smaller than the number of overall edges that are adjacent to the smaller partition.
This implies that even if we were allowed to choose the underlying $d$-regular graph structure, it would be impossible to construct a disassortative partition for such a large imbalance.

For a fixed $d$ and $m$, the entropy is decreasing as $H_{\rm ass}$ is increased and $H_{\rm dis}$ is decreased. 
Results with respect to different $d$ dependent on $H$ are shown in Figure~\ref{fig:vary-d}. 
%
%Here, the entropy decreases as the degree increases for $H=\lceil \frac{d}{2} \rceil$ and all $m$.
%This is consistent for all $H_{\rm dis} \geq d/2 + 1$ and likewise for all $H_{\rm ass} \leq d/2$.
%For the cases $H=\lceil d/2 \rceil \pm 1$ this is not the case, and already hints at the fact that the relevant scaling for large $d$ is not linear for $H$ when the respective problem becomes more difficult (more results are reported in the Appendix~\ref{app:extra-H-d}).
%
Table \ref{tab:results-entropy} then summarizes the values of degree $d$ and parameters $H$ for which assortative partitions exist for zero magnetization.

\begin{figure}
    \centering
         \includegraphics[width=0.49\textwidth]{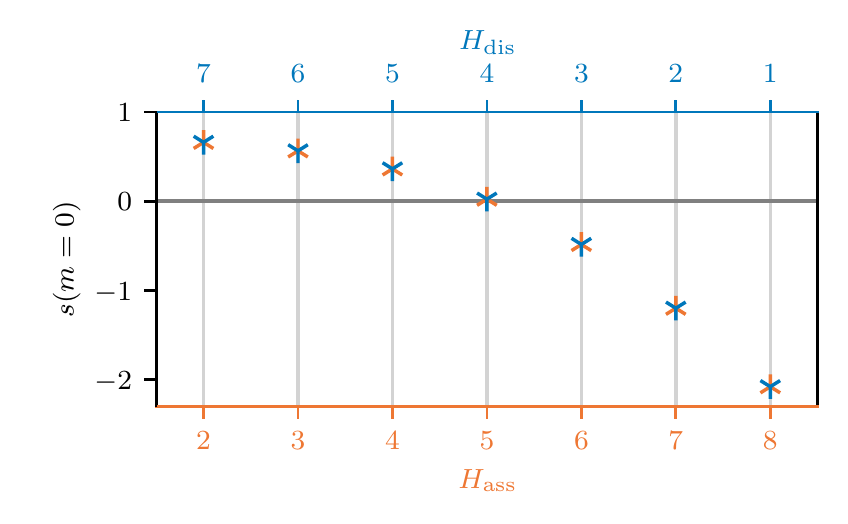}
    \caption{The RS entropy for $d=8$ and different values of $H$ for both assortative and disassortative partitions. 
    The magnetization is $m=0$ and the results are exactly the same for the corresponding reparameterization.}
    \label{fig:comp-happy-unhappy}
\end{figure}

\begin{figure}
     \centering
     \begin{subfigure}[b]{0.49\textwidth}
         \centering
         \includegraphics[width=\textwidth]{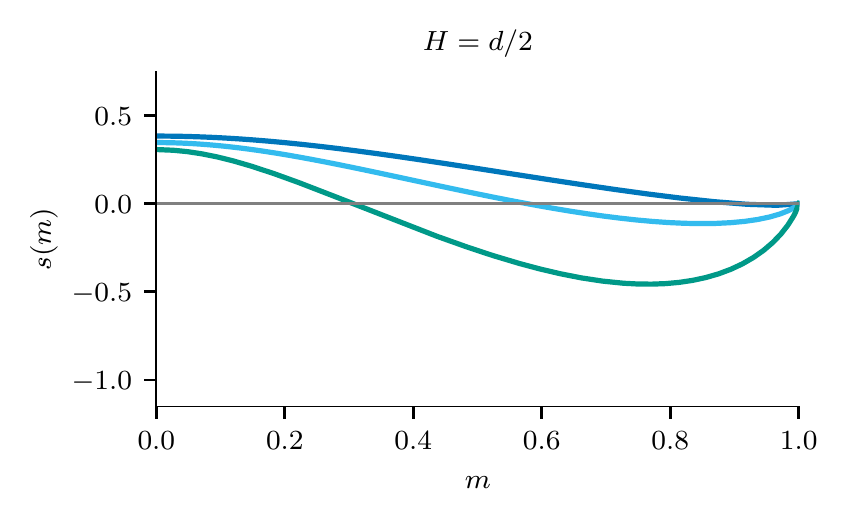}
           \caption{assortative}
        \label{fig:vary-d-ass}
     \end{subfigure}
     \begin{subfigure}[b]{0.49\textwidth}
         \centering
         \includegraphics[width=\textwidth]{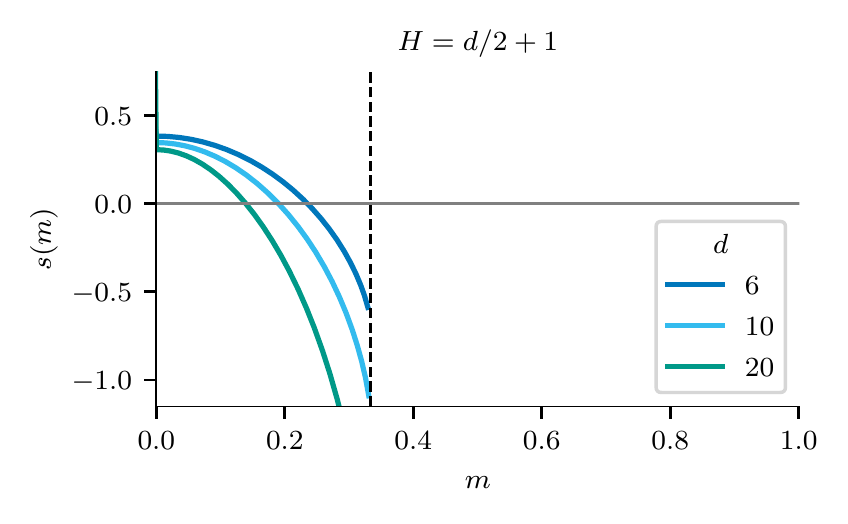}
           \caption{disassortative}
        \label{fig:vary-d-dis}
     \end{subfigure}
     \caption{\textit{Entropy $s(m)$ in the ($H,m$)-assortative partition problem on $d$-regular graphs.} The black dashed line for disassortative partitions shows the $\tilde{m}$ at which the exact number of minimally required edges of the bigger group to the smaller one is larger than the the number of edges the smaller group can maximally be connected to.
     } \label{fig:vary-d}
\end{figure}

%\paragraph{Equivalence between $H$-assortative bipartition to reparameterized disassortative partitions} While on a given single graph there is no clear relation between assortative and disassortative partitions, concerning their asymptotic properties, there is a close relation. 
%
%An example of this is shown in 

\begin{table}[ht!]
    \centering
    \small{
    \input{ferro-comparison-with_algo.tex}
}
    \caption{Results for the assortative partition problem for different $d$ around the relevant phase transitions, when the solution space becomes frozen at $d>H/2$.  
    Cases in the frozen phase are marked with bold face. The phase where typically no solutions exist (the entropy is negative) is marked in italic face.
    The algorithms were both run on the same 10 randomly sampled $d$-regular graphs of size $n=10,000$ and each run for at max $20,000$ iterations. 
    The minimum and median percentage of $H$-disassortative nodes in the resulting partition is reported.}
    \label{tab:results-entropy}
\end{table}

\subsection{Stability of the replica symmetric solution}
\label{sec:stability}

%\paragraph{Population dynamics for RS fixed point stability}
As anticipated, the Belief Propagation approach is not always exact even on large tree-like graphical models. A necessary condition for correctness of this approach is convergence of BP to a fixed point on a single full graph. Such a convergence can be checked by running BP on a given large graph initialized at the RS fixed point and perturbing the local messages slightly by some random noise~$\varepsilon$. 
%\todo{Motivate why we care about the stability, relation to small noise reconstruction}
%Several methods to check the stability of the RS fixed point have been used in the literature (for a summary, see ).
%One method consists of 

Since running the BP on large sets of $d$-regular random graphs is computationally heavy, we resort to approximating this behavior on a full graph using population dynamics \cite{zdeborovaStatisticalPhysicsHard2008a}.
This means, for updating an iterate at time $t$ we keep an array of $N$ messages $[\chi_i]_{i=1,...N}$. 
An update is constructed by randomly sampling $d-1$ neighbors for each new message in the array and applying the fixed point iteration \eqref{eq:final-fixed-point} on it.
We also adapt $\mu$ appropriately to reach the desired value of magnetization $m$.
If such an iteration of population dynamics, initialized by a noise-permuted version of the BP fixed point converges back, this fixed point is stable. 
Otherwise, it is unstable and the RS solution is not valid.

In running the population dynamics initialized at the BP fixed point with some small noise added, we checked the stability of this fixed point for $d\in[3,20]$ for bipartitions and for varying $m$ when $d \in [3,10]$.
In all cases the BP was stable.

\section{Asymptotic equivalence between the assortative and disassortative partition problem at \texorpdfstring{$m=0$}{m=0}}\label{sec:equality-explanation}

In the previous section we derived that on random regular graphs under the assumptions of the replica symmetry the entropy of balanced (i.e. $m=0$) assortative partitions is in the large size limit the same as the entropy of the disassortative partitions. 

One should note that on a single instance of a random regular graph the knowledge of a assortative partition does not give us any clear algorithmic way to obtain a disassortative partition, and vice versa. 
Thus the equality of the associated entropy should appear as a rather surprising property.
It is also worth noting that this property does not hold for non-balanced partitions $m\neq 0$ as shown in examples in Section~\ref{sec:bp-results}. 
We also do not expect this asymptotic relation between the two problems to hold on graphs with varying degree as in that case partitions of balanced size but not-balanced average degree could cause differences. 
We do not expect this property to be exactly true for the number of assortative and respectively dissasortative partitions on finite size graphs, but only for the entropies in the limit of large size. 

On the other hand we do expect and conjecture this equivalence to hold even if replica symmetry breaking needs to be taken into account in order to evaluate the value of the concerned entropy. We hence conjecture: 

\begin{conj}
    In the limit $n\to \infty$, on random $d$-regular graphs, the entropy of balanced assortative partitions at $H_{\rm ass}$ is equal to the entropy of dis-assortative partitions at $H_{\rm dis} = 1+d-H_{\rm ass}$.    
\end{conj}

This conjecture is reminiscent of an analogous one made in \cite{zdeborovaConjectureMaximumCut2010} for the size of the maximum cut and the size of minimum bisection on random $d$-regular graphs.
We will now present an argument why this conjecture is expected to hold beyond the replica symmetry assumption.
The assumptions we are making are common to the replica and the cavity methods of spin glasses \cite{mezard1987spin}.
This argument applies to $d$-regular random graphs, that are locally tree-like in the large size limit, and balanced $m=0$ partitions.
Similar to \eqref{eq:prob-dist} we use the following definition for the partition problem constraints 
$$
     \sum_{j \in \partial i} J_{ij} x_i x_j \ge 2H_{\rm ass} - d \quad \forall i\, , 
$$
where $J_{ij} = 1$ for all $(ij)\in E$ in the assortative case, and $J_{ij} = -1$ for all $(ij)\in E$ for the dis-assortative case taking $H_{\rm dis} = 1+d-H_{\rm ass}$.
We could even define a third \textit{spin-glass} version of the partition problem where $J_{ij} \in \{ \pm 1\}$ chosen at random for every edge $(ij)\in E$.
The difference between the assortative, the disassortative and the spin glass partitions hence boils down to the different choice of the parameters $J_{ij}$.

We now consider that whenever the cavity method arguments are used, there is always an underlying assumption that the associated probability measure, \eqref{eq:prob-dist} in our case, can be decomposed into extremal Gibbs measures, so called states, and each of these states is such that the so called point-to-set correlations decays to zero within the state, meaning that a root of a tree has no recollection of what is going on at a far away boundary if that boundary is chosen from the extremal Gibbs measure itself.
The assumptions of the existence of such a decomposition is at the core of the cavity/replica method.
With this in mind we now realize that the difference in $J_{ij}$'s that distinguishes between the assortative, the disassortative and the spin glass partitions can be pushed to the boundary of a rooted tree graph by applying the following change of variables: 
$$
      J_{ij} = \sigma_i  \tilde J_{ij} \sigma_j\,   \quad \quad 
      x_i  = \tilde x_i \sigma_i\, ,
$$
where $\sigma_i \in \{\pm 1\}$ are chosen in such a way that $\tilde J_{ij} = 1$ for all $(ij)$.
This can be easily done by choosing $\sigma$'s recursively from the root, since we assumed $m=0$ for the $x$-variables, and flipping all or a random portion of them by $\sigma_i$ will not change the magnetization of the $\tilde x$ variables. This part of the argument would fail if we wanted to make it for $m \neq 0$.
This is also where the requirement of regular graphs enters, because for graphs with inhomogeneous degrees the degrees could be correlated to the spin values and this correlation could again break the argument that the system before and after the change of variables will have the same properties.  
With the tilde-variables we transformed the problem on the tree into the assortative case.
The difference was pushed to the boundaries which, by assumptions of the method, does not influence the solution obtained from the cavity/replica method. 
All in all, we showed that as long as the problem is on random-regular graphs and has $m=0$ then the properties of partitions are the same for the assortative and disassortative and spin glass case. 
We saw this manifest for the disassortative and assortative partitions already in Section \ref{sec:bp-results}.
In the following section, we will see that the results we show for the partition problems for $d \to \infty$ were also discovered in the context of spin glasses.

\section{Large \texorpdfstring{$d$}{d} asymptotics for assortative bipartitions}\label{sec:large-d-derivation}
To make the relevant properties of the entropy visible at large $d \to \infty$ with $m=0$, we introduce $h$ such that $\minH = d/2 + \minh \sqrt{d}$.
Then in the zeroth order the entropy at this limit is given by 
\begin{align}
    \lim_{d \to \infty} \Phi_{\text{RS}}(h,d) =  \Phi_{\text{RS}}(h)= -2 c^2 + \log\left(\text{Erfc}\left[\sqrt{2}(c-h)\right]\right) + O(1/\sqrt{d})\label{eq:large-d-its}
\end{align}
under the condition that $c(h)$ fulfills the following fixed point equation
\begin{align}
    c =\frac{\exp(-2(h-c)^2)}{\sqrt{2\pi} \cdot \text{Erfc}\left[\sqrt{2}(h-c)\right]}\, .
\end{align}
This formula can be derived from the fixed point criterion in \eqref{eq:dreg-fixed-point} and the formula for entropy \eqref{eq:entropyyy} together with the assumption that the magnetization is fixed.
We show the detailed derivation in the Appendix.
For example, from this we can calculate that the transition from positive to negative entropy occurs at $\minh_{\max}=0.17566$.
The quality of the expansion for general $H$ is already very good at $d=1000$, as shown in Figure~\ref{fig:large-d-extra} and exactly agrees with the scaling evident in the simulations as shown in Figure~\ref{fig:large-d}.
Another interesting point is $h=0$, i.e. $H=d/2$, where the entropy is $s(m=0) = 0.19923$.
Note that both these numbers are known in the spin glass literature. The entropy $s(m=0) = 0.19923$ is known to be the entropy of single-spin-flip-stable states in the SK model of spin glasses \cite{brayMetastableStatesInternal1981}. The maximum gap $h_{\rm max}=0.17566$ was computed in \cite{trevesMetastableStatesAsymmetrically1988a} as the largest margin of stability of metastable states in a spin glass. This follows from the equivalence between the assortative, the disassortative and the spin-glass partitions discussed in Section \ref{sec:equality-explanation}. 

\begin{figure}[ht!]
     \centering
         \includegraphics[width=0.49\textwidth]{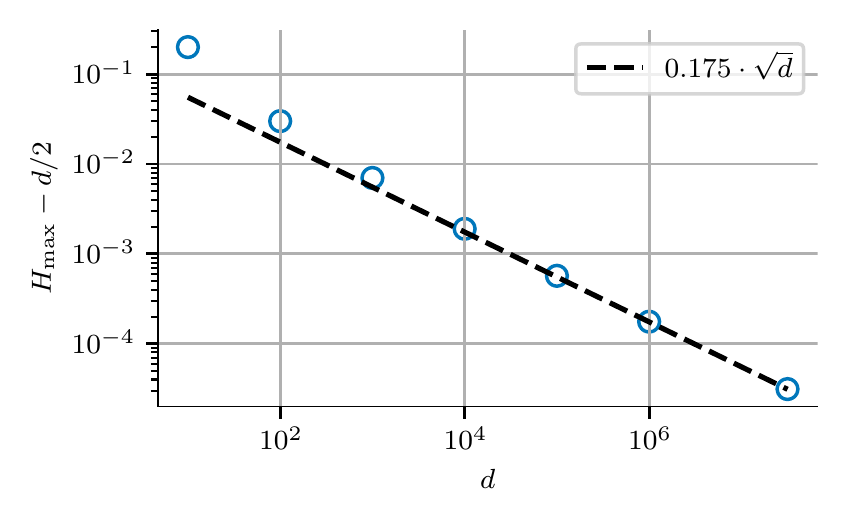}
        \caption{Large $d$ behavior of $H_{\max}$, the value of $H$ for which the entropy $s(0)$ first becomes smaller than zero. The gray dashed line is the constant derived analytically and agrees with the empirically calculated $H_{\max}$ as $d$ grows larger.}
        \label{fig:large-d}
\end{figure}

\begin{figure}[ht!]
     \centering
     \begin{subfigure}[b]{0.49\textwidth}
         \centering
         \includegraphics[width=\textwidth]{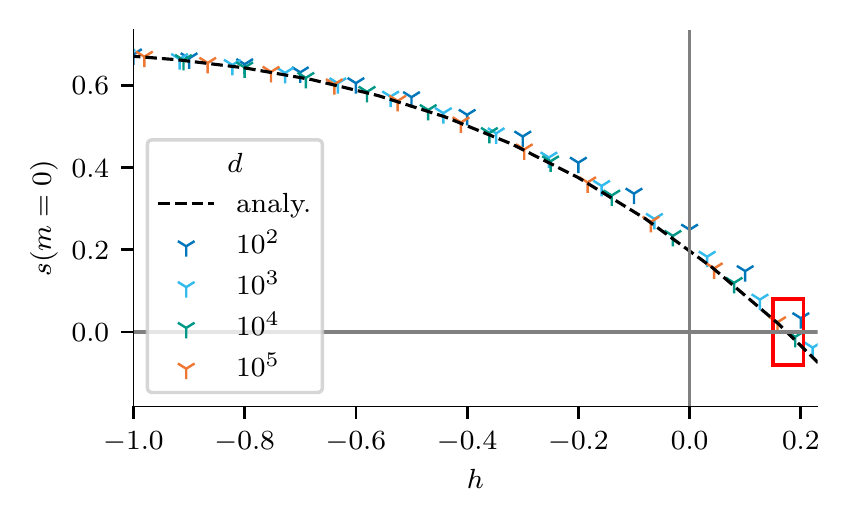}
     \end{subfigure}
     \hfill
      \begin{subfigure}[b]{0.49\textwidth}
         \centering
         \includegraphics[width=\textwidth]{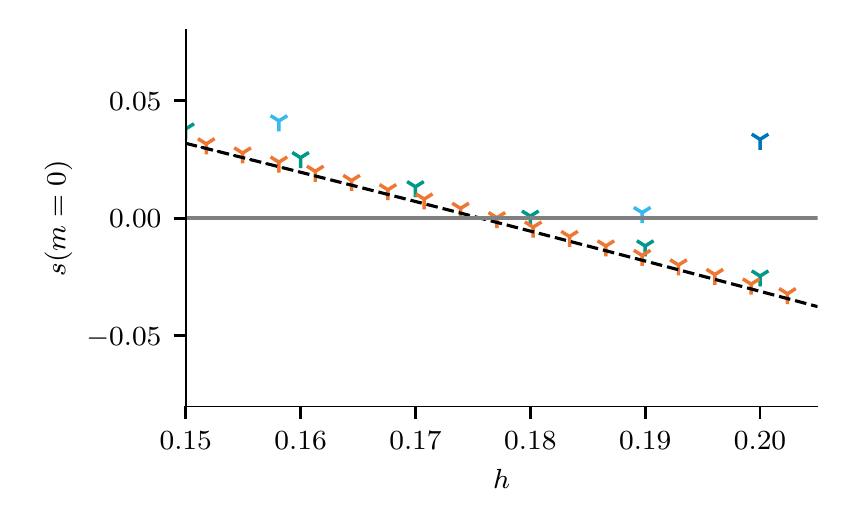}
     \end{subfigure}
        \caption{Large $d$ behavior of the entropy as a function of $h$, where $H=d/2+h \sqrt{d}$. Comparison of the analytic solution at $d \to \infty$ with empirical results of the fixed point iteration for large $d=10^2,10^3,10^4,10^5$. }
        \label{fig:large-d-extra}
\end{figure}

\section{Survey Propagation}\label{sec:survey-popagation-theory}

In Section \ref{sec:stability} we discussed the stability of the replica symmetric solution as a necessary condition for its correctness. But this is not a sufficient condition. It is predicted in statistical physics literature that the replica symmetric solution is asymptotically correct on tree-like factor graphs if the solution of the one-step replica symmetry breaking (1RSB) reduces back to it \cite{mezard2009information}. We thus investigate the 1RSB solution in this section for magnetization $m=0$.

By these means we can also resolve a problem from Section~\ref{sec:locked-frozen}, where we discussed to what extent the flip of a single node in a satisfying solution influences the assortativity of a fraction of other nodes in the graph for different values of $H$.
When an extensive fraction of the graph was directly influenced, we argued that this implies that valid solutions to the problem are isolated, i.e. spaced far apart from another. Such structure of the space of solutions is called the frozen 1RSB \cite{martinFrozenGlassPhase2004,aubinStorageCapacitySymmetric2019}. 
This would make it difficult to algorithmically discover them and manifest the computationally hard phase \cite{zdeborovaConstraintSatisfactionProblems2008,gamarnik2021overlap}. Because our previous argument from Section ~\ref{sec:locked-frozen} is inconclusive for cases when $H= d/2 +1$, we develop the 1RSB for our problem  also in order to settle whether the frozen 1RSB occurs in these cases.

The idea behind 1RSB is often described using the concept of reconstruction on trees \cite{zdeborovaConstraintSatisfactionProblems2008}.  
In reconstruction on trees, one considers one of the valid partitions on a given tree and then removes all information about the membership except at the boundary, the leaves of the tree.
The question is then, whether this boundary information is sufficient to determine something about which partition the root belonged to. A stronger requirement would be to determine the value of the root exactly based on the boundary, this is called the naive reconstruction in the literature \cite{zdeborovaConstraintSatisfactionProblems2008}. 
This can be answered by propagating constraints that the given leaves impose on their parents layer-wise up the tree until the root.
In some cases these implications may be inconclusive and exact reconstruction of the root might not be possible.
This happens, if there are many solutions with the same boundary but different values of the root.
If exact reconstruction is possible though, the information about the roots value is contained in the boundary.
Therefore, some variables between the root and the boundary must also be determined by the boundary, as otherwise the information would be lost.
These special determined variables are called frozen, as in all solutions that have the same boundary, these frozen variables must be the same.
%Vice versa, exactly those variables would change if the roots value was changed.
Computing the number of these frozen variables in the tree for some $H$ gives us information about how large the influence of a root flip is on the tree.

We can use this idea to compute the fraction of frozen variables on tree-like random graphs.
In this case, the naive reconstruction process can be approximated by the warning propagation algorithm (WP).
Given an initial boundary condition warning propagation finds the fixed point to which a reconstruction process on the full graph would converge to under the replica symmetric assumption.
The survey propagation algorithm then counts the logarithm of the number of those fixed points on random graphs containing frozen variables by running BP on the warning propagation constraints.
In the following we present both algorithms for random $d$-regular graphs and check if the result for the frozen 1RSB maps back to the replica symmetric solution from the BP to check that it is asymptotically correct.
This also tells us whether the fraction of frozen variables is extensive.

\paragraph{Warning propagation}
With warning propagation we can simulate the naive reconstruction process on random $d$-regular graphs.
Given initial membership assignments to partitions on the boundary we use the constraints of the (dis)assortative partition problem to infer which membership assignments on the remaining graph
can be freely assigned and which ones are frozen.
Recall that the variable $(x_i,x_j)$, represents the group memberships of the nodes $i$ and $j$ in either $V_-$ or $V_+$.
It can take on any of the four different values in $\{\pm 1 \}^2$, representing crossing and non-crossing edges between or within the two groups.
In this setting, a warning $w^{i\to j}$ is a message from $i$ to $j$ that tells node $j$ which values of the variable $(ij)$ are acceptable for $i$.
An assignment is acceptable for $i$ if it allows $i$ to be $H$-assortative given the warnings incoming to $i$ from all its other neighbors $k \in \partial i \setminus j$.
For example, a warning that node $i$ sends to $j$ could be $w^{i \to j}=\vcenteredinclude{\includegraphics[width=0.06\textwidth]{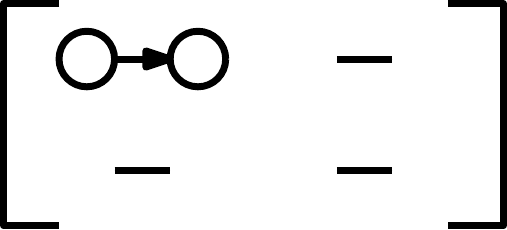}}$  (the color white denotes an assignment to partition $V_-$ and black to $V_+$) .
This tells $j$ that the only acceptable value of $(x_i,x_j)$ for $i$ is  \vcenteredinclude{\includegraphics[width=0.03\textwidth]{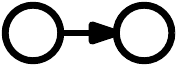}}.
Such a situation could occur when $i$ has to have at least $H$ neighbors in its own group, but it has only $H-1$ neighbors from its own group in the neighbors $\partial i \setminus j$.
Then the only way to fulfill this constraint is to require $j$ to be in the same partition as itself. 
Thus, $j$ is an example of a frozen variable since it cannot be chosen freely.
On the other hand, if $i$ does not have a specific requirement for $j$, for example because it already has at least $H$ neighbors in its own group without counting $j$, it sends the \textit{don't care}-warning to $j$, which we denote as $w^{i \to j}= \vcenteredinclude{\includegraphics[width=0.06\textwidth]{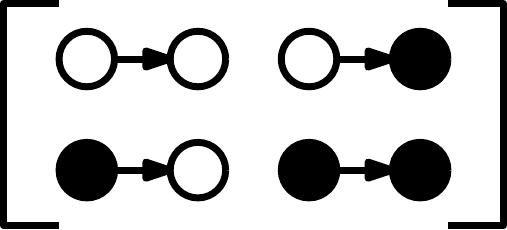}}$.
In this case $j$ can be chosen freely, i.e. it is not frozen.
While this symbolic notation is instructive, in the following we use the more succinct notation where $w^{i \to j} \in \smat{\pm}{\pm}{\pm}{\pm}$, so that a $w^{i \to j}_{x_i,x_j} = +$ denotes that the corresponding assignment $(x_i,x_j)$ is allowed and $=-$ that it is not.
Using this notation we have that $\vcenteredinclude{\includegraphics[width=0.06\textwidth]{survey-propagation.pdf}} = \smat{+}{+}{+}{+}$ and $\vcenteredinclude{\includegraphics[width=0.06\textwidth]{survey-propagation-only-one.pdf}}=\smat{+}{-}{-}{-}$.
Since $i$ can select any subset of the four possible assignments $\{\pm 1\}^2$, the warnings can take on 16 different values.
Note that the ``don't care" warning is special, as it is the only warning that allows $j$ to choose its assignment freely.\\
Under these conditions, we can compute the warning $w^{i \to j}$ given the warnings that $i$ receives from all other neighbors as
\begin{align}
    w^{i \to j}_{x_i,x_j} =  \begin{cases}
    +,& \text{if } \exists \{x_k\}_{k \in \partial i \setminus j}:  \mathcal{C}_H\left( \sum_{j \in \partial i} J_{ij} x_i x_j \right)=1\text{ and  }  w^{k \to i}_{x_k,x_i} {\small = +} \, ,\\
    -,              & \text{else}  \, . 
\end{cases}\label{eq:update-wp}
\end{align}
Note that by assigning $J_{ij}$ correctly we can obtain the warning propagation for either the assortative or disassortative case as in \eqref{eq:prob-dist}. 
It allows only those membership assignments $(x_i,x_j)$ for $(ij)$, for which $i$ could make itself $H$-(dis)assortative by selecting the appropriate allowed assignments from the warnings it received from its neighbors $k \in \partial i \setminus j$.
The first part of the constraint ensures that those assignments are allowed in the warnings $i$ received from those neighbors.
The $\mathcal{C}_H$ as defined in \eqref{eq:constraint} ensures that the local constraint on $i$ for (dis)assortativity is fulfilled.
For brevity of notation we will denote the warning propagation update as 
\begin{equation}
    w^{i \to j}_{x_i,x_j} = {\cal F} (\{ w^{k \to i}_{x_k,x_i}\}_{k \in \partial i\setminus j} ) \, . 
\end{equation}

On a tree with the leaves set to ``conclusive'' warnings where only one assignment is allowed, iterating \eqref{eq:update-wp} up until the root would allow us to exactly tell which values the root would be allowed to take on. Another relevant initialization in the context of frozen 1RSB is one where small noise is introduced on the boundary in the form of a small fraction of ``don't care" warnings. If such an initialization still implies the value of the root, we speak of the \textit{small noise reconstruction} \cite{zdeborovaConstraintSatisfactionProblems2008}, which then implies the presence of frozen 1RSB where an extensive fraction of nodes needs to change in a tree-like graph upon the flip of one node . 
%For random tree-like graphs and nodes initialized with conclusive warnings, iterating the update rule \eqref{eq:update-wp} until convergence as described for the BP on the full graph in Section \ref{sec:rs}, would reconstruct the solution if there is enough information available at initialization.
%Additionally, any valid solution must be a fixed point of this update rule.

\paragraph{Survey propagation}
%Warning propagation only gives the fraction of frozen variables for a specific instance of a graph, but here we consider random graphs.
We can use survey propagation to count the number of warning propagation fixed points and their entropy for random graphs.
A survey $\eta^{i \to j}$ is a probability distribution over the 16 different warnings that $i$ could send to $j$.
This is analogous to how we derived the belief propagation equations and the associated free entropy in Section~\ref{sec:rs-1}.
SP messages are also probability distributions over variable assignments and therefore $\Phi_{\rm SP}$ allows us to recover the entropy of the warning propagation solutions/fixed points by the same arguments.

For example, $\eta^{i \to j}\smat{+}{+}{+}{+}$ is the probability that $i$ sends a ``don't care" warning to $j$.
The surveys can then be updated as
\begin{align}\label{eq:32}
    \eta^{i \to j}_{w^{i\to j}} = \frac{1}{Z^{i \to j}} \sum_{\{w^{k\to i}\}_{k \in \partial i}} \ind\left(w^{i \to j}_{x_i,x_j} = {\cal F} (\{ w^{k \to i}_{x_k,x_i}\}_{k \in \partial i\setminus j} )      \right) \prod_{k \in \partial i} \eta^{k \to i}_{w^{k\to i}} \, , 
\end{align}
where the normalizing constant $Z^{i \to j}$ ensures that $\eta$ is a probability distribution.
At a fixed point of this iteration, one can compute the entropy of the warning propagation fixed point, that is referred to as the complexity $\Phi_{SP}$ in the literature,  as
\begin{align}
    n\Phi_{\text{SP}} &= \sum_{i \in V}\log(Z_i) - \sum_{(ij) \in E}\log(Z_{ij})\, , \\
    Z_i &= \sum_{\{w^{k\to i}\}_{k \in \partial i}} \ind\left(w^{i \to j}_{x_i,x_j} = {\cal F} (\{ w^{k \to i}_{x_k,x_i}\}_{k \in \partial i\setminus j} ) \right) \prod_{k \in \partial i}\eta^{k\to i}_{w^{k\to i}}\, , \label{eq:huhuhu}\\
    Z_{ij} &= \sum_{\{w^{i\to j},w^{j\to i}\}} \ind\Big(\exists (x_i,x_j): w^{i\to j}_{x_i,x_j} = w^{j \to i}_{x_j,x_i} = + \Big)\eta^{i \to j}_{w^{i\to j}} \eta^{j \to i}_{w^{j\to i}}\, . 
\end{align}
Some examples of how the constraint in \eqref{eq:huhuhu} for the specific case of $H$-disassortative partitions is resolved are shown in Figure~\ref{fig:explain-compatability}.
In analogy to the BP in Section~\ref{sec:BP-dreg} we can easily derive simplified update rules for random $d$-regular graphs under the assumption that $\eta = \eta^{i\to j}$ $\forall i,j$ and apply it to some initial distribution $\eta_0$.

\begin{figure}[ht!]
    \centering
         \includegraphics[width=0.24\textwidth]{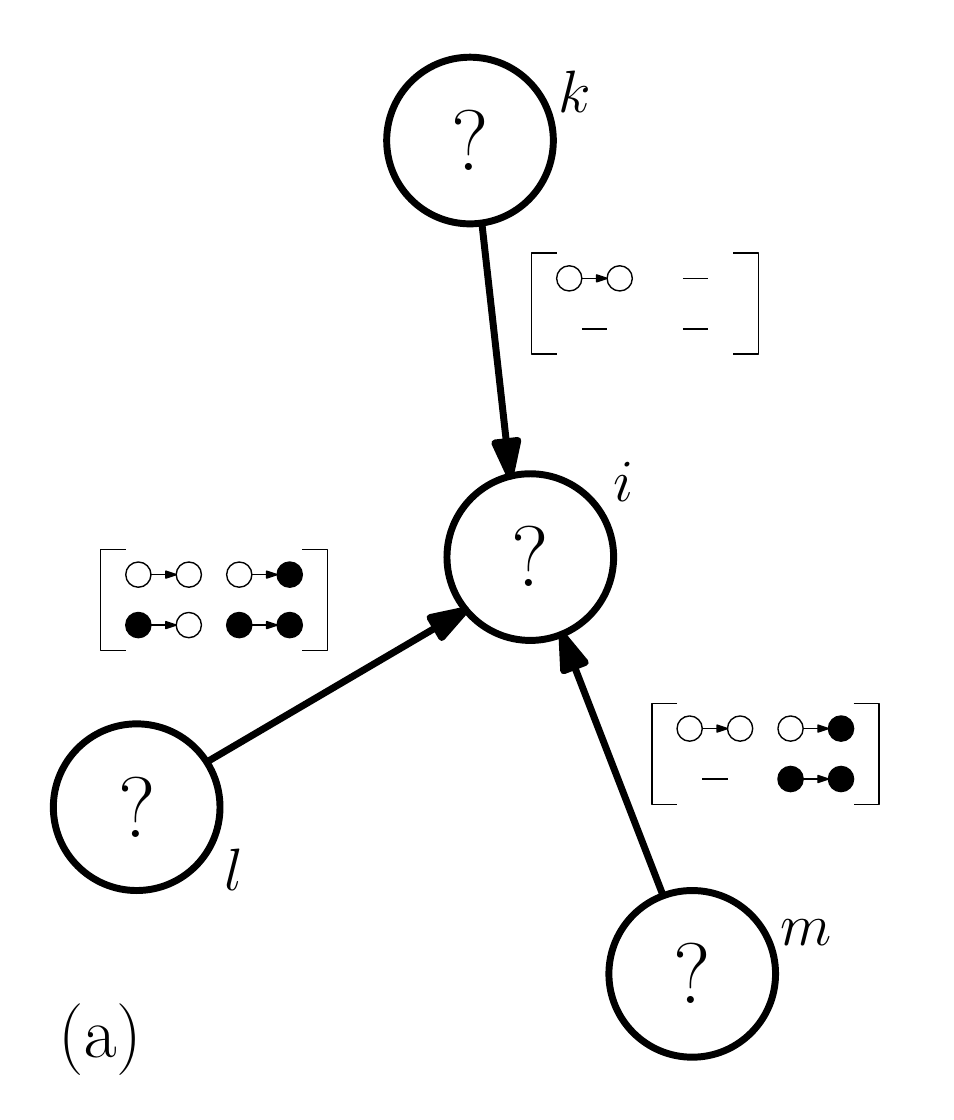}
         \includegraphics[width=0.24\textwidth]{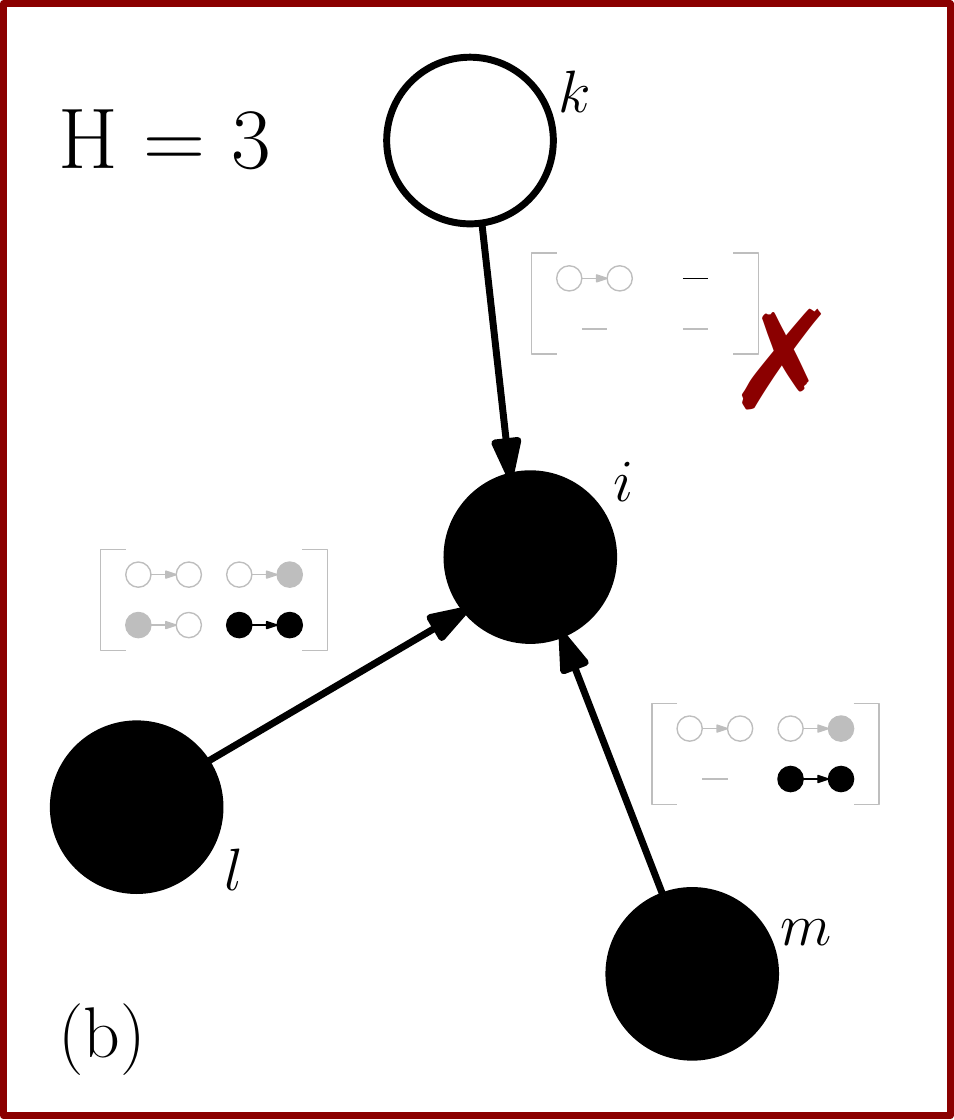}
         \includegraphics[width=0.24\textwidth]{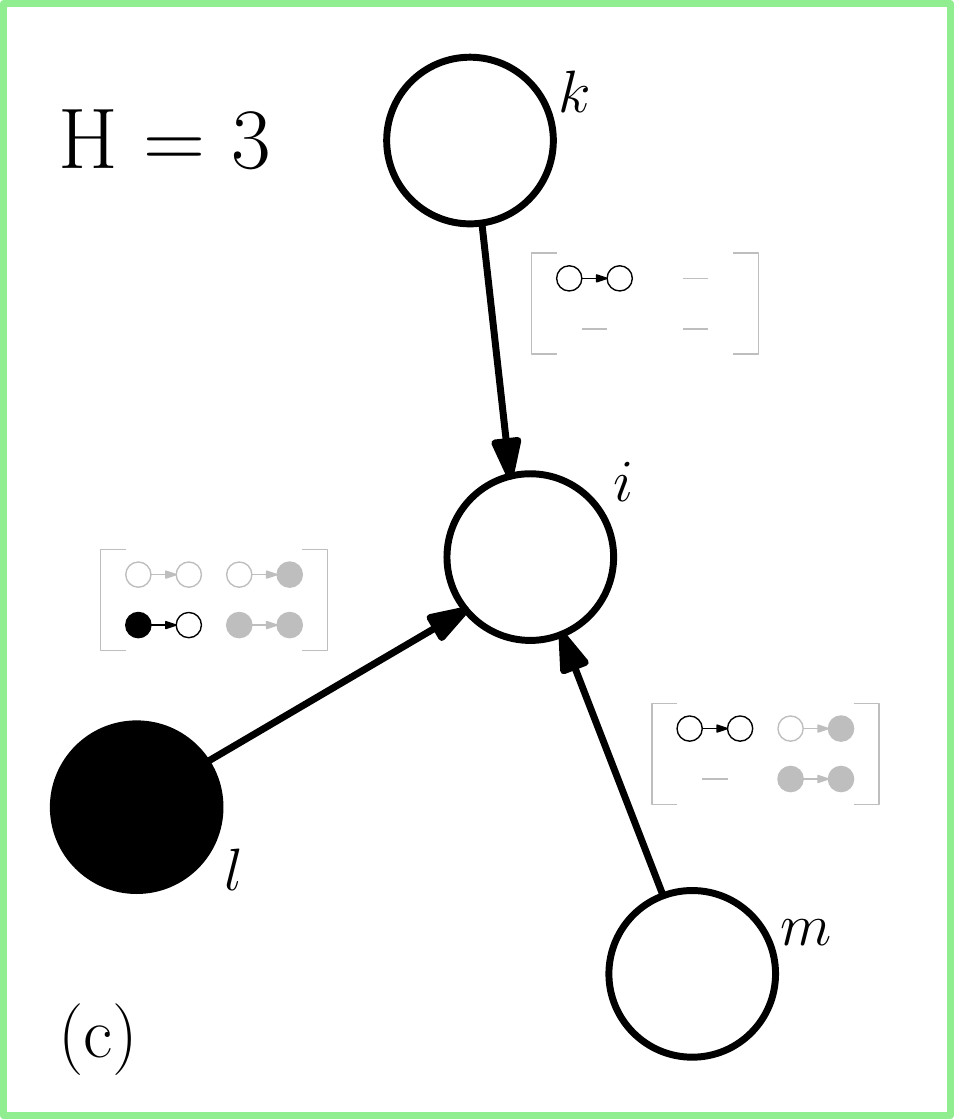}
         \includegraphics[width=0.24\textwidth]{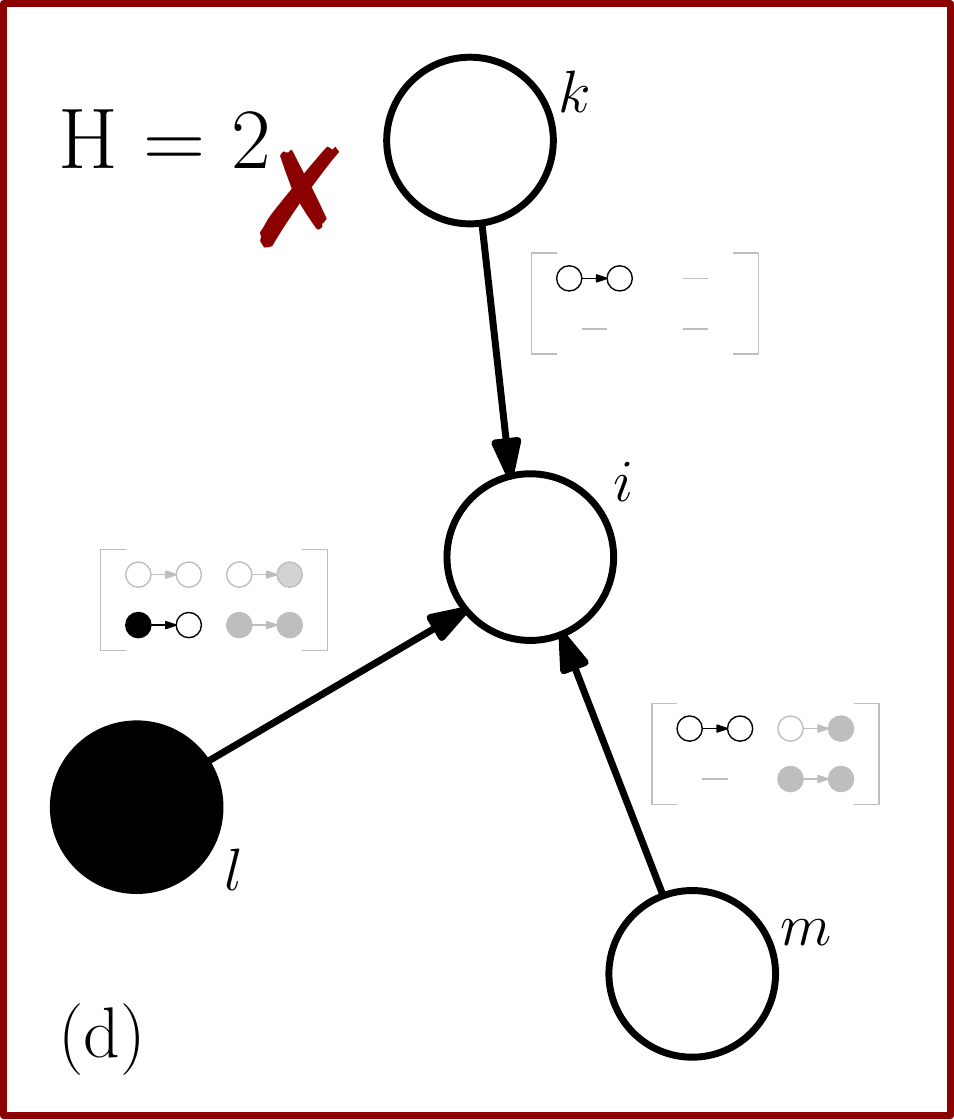}
    \caption{Illustration for finding a compatible membership assignment for an $H$-disassortative partition under the given incoming warnings $\{w^{j\to i }\}_{j \in \partial i}$ in \eqref{eq:huhuhu}. \textit{(a)} Incoming warnings to center node $i$ on a 3-regular graph, where the nodes have not been assigned to groups.
    \textit{(b)} The edge $(x_k,x_i) =$ \vcenteredinclude{\includegraphics[width=0.02\textwidth]{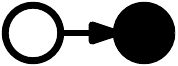}} is not compatible with the incoming warning $w^{k \to i}$.
    \textit{(c)} All membership assignments are compatible with the incoming warnings and the center node $i$ is $3$-disassortative. 
    \textit{(d)} The membership assignments are compatible with the warnings but the center node $i$ is not disassortative for $H=2$.
}
    \label{fig:explain-compatability}
\end{figure}

\paragraph{At $H \leq \lceil\frac{d}{2}\rceil$ the solution space of the disassortative partition problem is frozen 1RSB}
Setting $\eta_0$ to correspond to the small noise reconstruction, i.e. setting a small probability for the ``don't care" survey and the rest of the warnings in the survey to a conclusive assignment, we can probe the fraction of non-frozen variables at the convergent distribution $\eta^*$ as $\eta^*\smat{+}{+}{+}{+}$, as well as the entropy of the warning propagation's fixed points $\Phi_{\rm SP}$.

In running the SP equations \eqref{eq:32} for $d$-regular graphs, we consistently observe that for the disassortative case and for $H \le \lceil\frac{d}{2}\rceil$ only exactly four of the 16 possible warnings have a non-zero probability at convergence.
Two of them are $\smat{-}{-}{\vcenteredinclude{\includegraphics[width=0.02\textwidth]{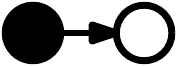}}}{\vcenteredinclude{\includegraphics[width=0.02\textwidth]{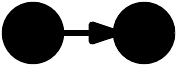}}}$ and $\smat{\vcenteredinclude{\includegraphics[width=0.02\textwidth]{survey-propagation-zero-zero.pdf}}}{\vcenteredinclude{\includegraphics[width=0.02\textwidth]{survey-propagation-zero-one.pdf}}}{-}{-}$ which require the membership of the outgoing node  to be fixed to a specific group but where the membership of the incoming node can be chosen freely.
The other two are $\smat{-}{-}{\vcenteredinclude{\includegraphics[width=0.02\textwidth]{survey-propagation-one-zero.pdf}}}{-}$ and $\smat{-}{\vcenteredinclude{\includegraphics[width=0.02\textwidth]{survey-propagation-zero-one.pdf}}}{-}{-}$ which both require that the edge itself be a specific crossing edge.
Since at convergence there is zero probability of the ``don't care" edge $\eta\smat{+}{+}{+}{+}$ for $H \le \lceil\frac{d}{2}\rceil$, the $H$-disassortative bipartition problem is frozen in this case.
Then, there is also a direct mapping between the RS fixed point $\chi$ and the survey propagation fixed point $\eta$ for a $H$ according to
\begin{align}
    \chi_{\scriptscriptstyle (+1,+1)} = \chi_{\scriptscriptstyle (-1,-1)} &= \frac{\eta\smat{+}{+}{-}{-}}{2 \cdot \eta\smat{+}{+}{-}{-}  +1}\, ,\\
    \chi_{\scriptscriptstyle (-1,+1)} = \chi_{\scriptscriptstyle (+1,-1)} &= 1/2 -\chi_{\scriptscriptstyle (-1,-1)}\, .
\end{align}
Under this mapping, the BP and SP entropy are exactly the same as shown in Table~\ref{tab:results-entropy}.
In this case of frozen 1RSB, we can conclude that the RS solution based on analyzing Belief Propagation is correct.

For $H > \lceil\frac{d}{2}\rceil$ the SP converges to a fixed point where every warning is ``don't care''. 
Therefore the reconstruction is not possible and we are not in the frozen 1RSB.

An example for these two different behaviors for $d=18$ is shown in Figure~\ref{fig:sp-example}, and this is consistent with the general picture for smaller or larger $d$.

By the equivalence between the (dis)assortative partitions at $m=0$, we can conclude that the same holds for the assortative partitions: in this case, the frozen variables\footnote{For assortative partitions, we have that only $\smat{-}{-}{+}{+}$, $\smat{+}{+}{-}{-}$, $\smat{+}{-}{-}{-}$ and $\smat{-}{-}{-}{+}$. The difference here is that instead of the crossing edges for the disassortative partitions, only warnings which are fixed to non-crossing edges appear.} appear at $H > \lceil\frac{d}{2}\rceil$.
This completes the phase diagram for magnetization zero.

\begin{figure}[ht]
\centering
         \includegraphics[width=0.54\textwidth]{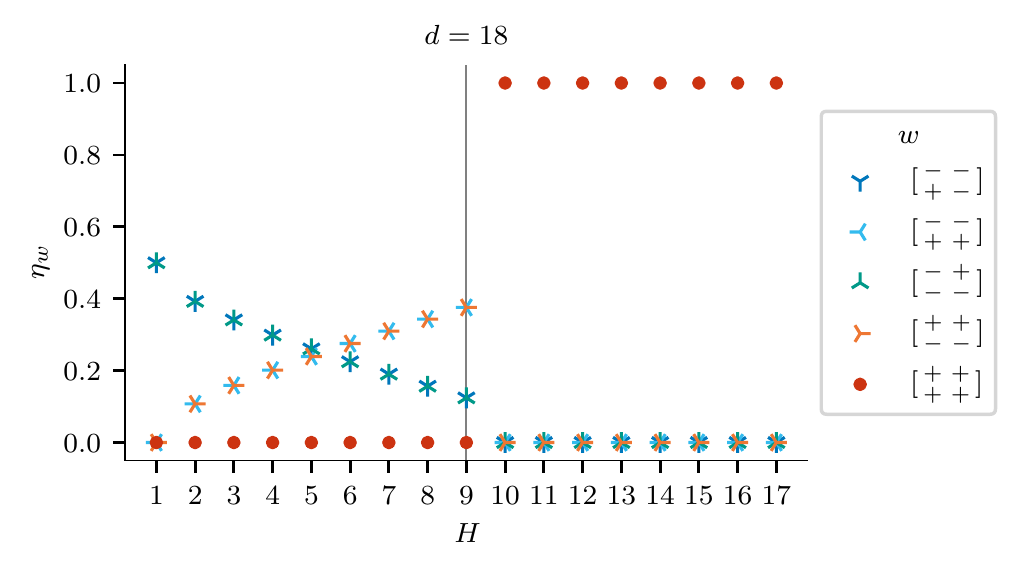}
    
    \caption{Fixed points $\eta_w$ of the survey propagation equations for disassortative partition with $d=18$.
    All warnings $w$ that are not explicitly shown in the plot are zero. 
    Note that the symmetry between the partitions is present in the fixed point: The values of $\eta\smat{-}{-}{+}{-}$ and $\eta\smat{-}{+}{-}{-}$ are equivalent and likewise for  $\eta\smat{-}{-}{+}{+}$ and $\eta\smat{+}{+}{-}{-}$.
    The \textit{don't care} warning (red) indicates the fraction of warnings that are not frozen.
    This clearly shows that the problem is frozen for $H_{\rm dis} \leq 9$, when the probability of such a warning is zero, and otherwise there are no frozen variables at all.
    }
    \label{fig:sp-example}
\end{figure}

\section{Algorithmic Performance at \texorpdfstring{$m=0$}{m=0}}\label{sec:algorithms}

We now take a look at two algorithms that aim to find a solution on a specific instance of the assortative $H$-partition problem at magnetization zero.
To avoid algorithms from converging to the ```all $+1$'' and ``all $-1$'' partitions, we consider only those that keep the magnetization constant during their run by using SWAP operations between the two partitions.

\paragraph{Algorithm EagerSWAP: Swapping random disassortative nodes} The graph is initialized with a random but balanced partition.
For every node it is calculated whether it is assortative or disassortative depending on the parameter $H$, i.e. if the constraint is fulfilled given their neighbors memberships.
Given this information, two diassortative nodes are selected at random, one from group $V_-$ and the other from group~$V_+$.
Then, they become members of the opposite group and all affected nodes in the graph are updated.
These steps are repeated until the number of assortative nodes in the graph no longer fluctuates, or when the maximum number of iterations is reached.

\paragraph{Algorithm GreedySWAP: Swapping most disassortative nodes \cite{ferberFriendlyBisectionsRandom2021a}}

Define the \textit{assortativeness} of node $i$ as the number of neighbors with membership $x_i$ minus the number of neighbors with a different membership:
\begin{align}
    \text{assortativeness}(i) := \sum_{j\in\partial i} \ind(x_i = x_j) - \sum_{j \in \partial i}   \ind(x_i \neq x_j)\, . 
\end{align}
Upon a random and balanced initial partition, the most disassortative node of each partition is selected.
Then, these two nodes swap their groups.
Afterwards, the assortativeness of all effected nodes is updated and the same process repeated.
Analogous to EagerSWAP, the current partition is returned once the assortativeness of the overall nodes is stable or the maximum iterations are reached.

\paragraph{Results}
In both cases, the balance between the partition sizes is always maintained during the iterations.

Note that GreedySWAP was proven to produce with high probability $n-o(n)$ nodes with assortativeness $\geq 0$ for random Erdős–Rényi graphs with degree $\frac{n}{2}$ and $n$ nodes \cite{ferberFriendlyBisectionsRandom2021a}.
This means that for $n-o(n)$ nodes at least half of their neighborhood is in their own group.
This is similar the problem of $\lceil \frac{d}{2} \rceil$-assortative partitions that we cover in this paper, although we do not have a random, but fixed degree $d$ and consider sparse and not dense graphs.

%\subsection{Performance}
The algorithms are run on randomly generated $d$-regular graphs and the resulting number of disassortative nodes at convergence is recorded. 
If there are no disassortative nodes in this result, a solution to the $H$-assortative graph problem was recovered.
Because the algorithm GreedySWAP does not depend on the parameter $H$, for every instance with degree $d$ it is only run once.
Then, the number of disassortative nodes of the recovered partition is evaluated for different $H$.

The results of several runs of these algorithms are shown in Table~\ref{tab:results-entropy}.
GreedySWAP finds satisfying partitions for $H \leq d/2$ more reliably than EagerSWAP, which sometimes leaves a small number of nodes disassortative.
For odd $d$, where $H = \lceil d/2 \rceil$ both algorithms successfully find partitions that leave only a small number of nodes disassortative.
In the hard phase where we have a negative $s(0)$, both algorithms fail to find good solutions.
In the frozen phase, the algorithms both do not find completely satisfying solutions, but EagerSWAP finds solutions that have more assortative nodes than GreedySWAP.

Since in Table~\ref{tab:results-entropy} we only give results for fixed $n=10,000$, we also check finite-size effects in Figure~\ref{fig:finite-size}.
For the example of $d=8$, the percentage of disassortative nodes scales as a power law when $H=4$, which is in line with our prediction that in the easy phase the algorithm should be able to find solutions when $n \to \infty$.
In contrast, the percentage of disassortative nodes plateaus in the frozen-1RSB phase when $H=5$, which corroborates our prediction that in this setting solutions are difficult to find as the percentage plateaus despite the growth of $n$. 

Additionally, both algorithms can be adapted equivalently for the disassortative bipartition.
For EagerSWAP, the criterion changes to swapping random assortative nodes.
Similarly, for GreedySWAP, instead of swapping the most disassortative nodes, the most assortative ones are swapped.
Under these adaptations we recover the same behavior as for the assortative case.
In the disassortative easy phase we find such partitions easily whereas in the the hard phase the solutions are not found.

\begin{figure}[hbt]
    \centering
         \includegraphics[width=0.49\textwidth]{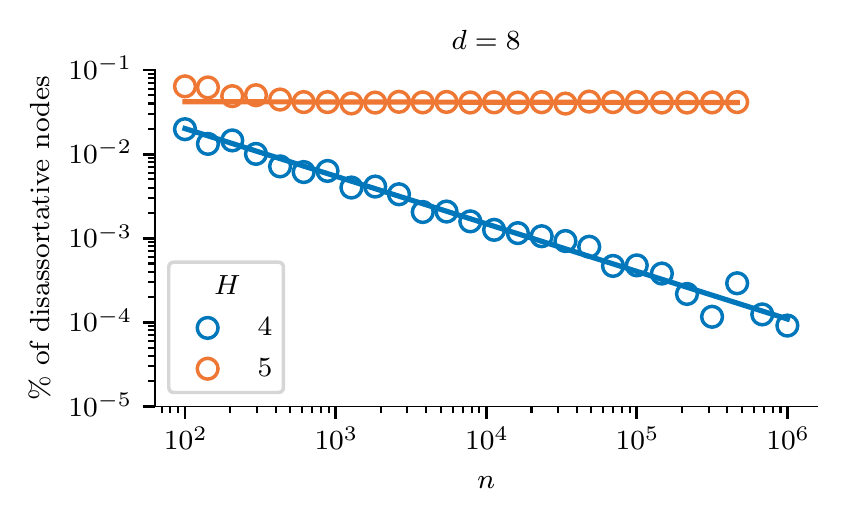}
         
    \caption{Mean remaining percentage of disassortative (unsatisfied) nodes in partitions found by EagerSWAP for $d=8$ for an increasing graph size $n$.
  In the predicted easy phase at $H=4$ this percentage decreases as a power law and the algorithm finds increasingly good solutions as $n \to \infty$.
    In the frozen phase at $H=5$, which should be difficult for local search algorithms,  the percentage of disassortative nodes plateaus which is in line with our prediction.
    From $n=10^2$ to $n=10^5$ the results are averaged over 50 runs, for larger instances the experiment was repeated 10 times. The maximum number of iterations was $10\cdot n$. }
    \label{fig:finite-size}
\end{figure}

\section{Acknowledgements}

We thank Afonso Bandeira for bringing this problem to our attention. 
This work started as a part of the doctoral course Statistical Physics For Optimization and Learning taught at EPFL in spring 2021.
We thank Markus Müller, Martin Loebl, and Florent Krzakala for discussions and pointers to some of the existing literature related to this work. 
We acknowledge funding from the ERC under the European Union’s Horizon 2020
Research and Innovation Programme Grant Agreement 714608-SMiLe.

\newpage

\printbibliography
\newpage
\appendix
\section{Appendix}

%\subsection{RS/1RSB calculation tricks}
 %Dynamic programming approach\\
 %Implementing the delta function by only selecting fixed points\\
 %keeping an array of messages for the population\\
 %dampening by only keeping a subsample

\subsection{Derivation of large \texorpdfstring{$d$}{d} asymptotics}
\begin{figure}[b]
    \centering
    \includegraphics{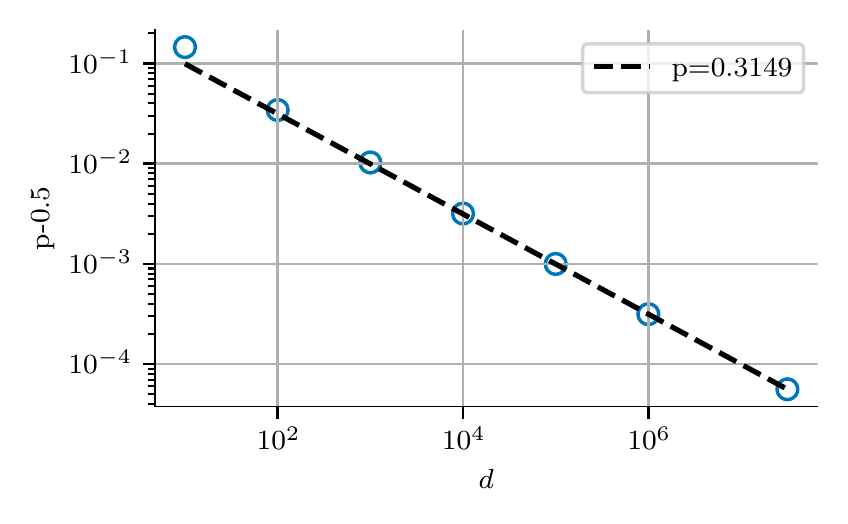}
    \caption{Numerical results for the scaling of the fixed point $p$ as $d$ grows larger.}
    \label{fig:fixed-poin-p}
\end{figure}

We can simplify the fixed point equation \eqref{eq:dreg-fixed-point} further 
taking into account the symmetry between the two partitions for magnetization $m=0$
\begin{align}
    a = \chi_{0,0} &= \chi_{1,1}\, ,\\
    b = \chi_{0,1} &= \chi_{1,0}\, .
\end{align}
We also know from the fact that  $\chi$ is a probability distribution, that we have $ ||\chi||_1 = 1$ and $\chi_{0,0}, \chi_{1,1}, \chi_{0,1},  \chi_{1,0} \geq 0$ so we have
\begin{align}
    b &= 0.5 - a\, .
\end{align}
Defining $p = 2a$, we get a simplified fixed point equation in terms of the pdf ($A_1$) and cdf ($A_2$) of a Binomial distribution with $d-1$ trials and success probability $p$
\begin{align}
    A_1 &= \sum_{i=\minH}^{d-1} \binom{d-1}{ i} p^i (1-p)^{d-1-i}\, , \label{eq:a1}\\
    A_2 &=  \binom{d-1}{ \minH - 1} p^{\minH - 1} (1-p)^{d-\minH}\, , \label{eq:a2}\\
    p &=  \frac{A_1+A_2}{2 A_1 + A_2}  = \frac{1 + \frac{A_1}{A_2}}{1 + 2 \frac{A_1}{A_2}}\, .\label{eq:finalfixedddd}
\end{align}
Using the same variable definitions, we find an expression for the entropy from  \eqref{eq:Entropy-RS} that reads
\begin{align}
    \Phi_{\text{RS}} = \log\Big[\sum_{r=H-1}^d \binom{d}{ r} p^r (1-p)^{(d-r)}\Big] + (1 -d ) \log(2) - \frac{d}{2} \log\Big[2((\frac{p}{2})^2+(\frac{1-p}{p})^2)\Big]\, . \label{eq:finalent}
\end{align}
We now make two assumptions that come from the numerical evaluation of \eqref{eq:dreg-fixed-point}, see Figure \ref{fig:large-d} and \ref{fig:fixed-poin-p}.
We found that the relevant scaling for the first order occurs at $\sqrt{d}$, so we assume
\begin{align}
p &= \frac{1}{2} + c \sqrt{\frac{1}{d}}\, ,\label{eq:as1} \\
    \minH &= \frac{d}{2} + h \sqrt{d}\, ,\label{eq:as2}
\end{align}
where $c,h \in \mathbb{R}$. 
In the following we use $\bar{d} = d-1$ to lighten the notation slightly.

For our derivation of the first order expansion we use Stirling's approximation for the binomial coefficients in \eqref{eq:finalent} and \eqref{eq:a2}.
For the median $k$ of a discrete probability distribution it holds that $\text{Prob}(X \leq k) \leq 1/2$ and $\text{Prob}(X \geq k) \leq 1/2$.
We use the fact that for a Binomial distribution  $B(n=\bar{d},p)$, the value $\bar{d}p$ is a good approximation of the median $k$.
Therefore, the cumulative distribution function $P(K \geq p\bar{d}) = \sum_{i=p\bar{d}}^{\bar{d}} \binom{\bar{d}}{ i} p^i (1-p)^{\bar{d}-i} \sim 1/2$. 
Also we empirically observe, that $p\bar{d} > h$, so we split the sum in $A_1$ in  \eqref{eq:a1} at $p\bar{d}$. 
\begin{align}
    A_1 &= \sum_{r=H}^{\bar{d}} \binom{\bar{d}}{ r} p^r(1-p)^{\bar{d}-r} \\
    &\sim \frac{1}{2} + \sum_{r=H}^{\bar{d}p} \binom{\bar{d}}{ r} p^r (1-p)^{\bar{d}-r}\\
    \intertext{we use Stirling's approximation to approximate the binomial coefficient}
    &\sim \frac{1}{2} + \sum_{r=H}^{\bar{d}p} \sqrt{\frac{n}{2 \pi r (\bar{d}-r)}} \frac{\bar{d}^{\bar{d}}}{r^r (\bar{d}-r)^{(\bar{d}-r)}} p^r (1-p)^{\bar{d}-r}\, .\\
    \intertext{We substitute the sum for an integral as}
    &\sim \frac{1}{2} + \int_{H}^{p\bar{d}}  \sqrt{\frac{n}{2 \pi r (\bar{d}-r)}} \frac{\bar{d}^{\bar{d}}}{r^r (\bar{d}-r)^{(\bar{d}-r)}} p^r (1-p)^{\bar{d}-r} dr\\\intertext{and insert our assumptions on the scaling of $p$ and $H$ from \eqref{eq:as1} and \eqref{eq:as2}}
    &\sim \frac{1}{2} + \int_{d/2+h \sqrt{d}}^{d/2+c \sqrt{d}}  \sqrt{\frac{n}{2 \pi r (\bar{d}-r)}} \frac{\bar{d}^{\bar{d}}}{r^r (\bar{d}-r)^{(\bar{d}-r)}} p^r (1-p)^{\bar{d}-r} dr\\
    \intertext{putting the boundaries of the integrals inside we obtain}
    &= \frac{1}{2} + \int_{h}^{c} 2^{d+\frac{1}{2}} \bar{d}^{\bar{d}} d \left(-2 \sqrt{d} r+d-2\right)^{\sqrt{d} r-\frac{d}{2}+1} \left(2 \sqrt{d} r+d\right)^{-\sqrt{d} r-\frac{d}{2}}\\
    & \sqrt{\frac{d-1}{-4 \pi d r^2-4 \pi \sqrt{d} r+\pi (d-2) d}}  \left(\frac{1}{2}-\frac{c}{\sqrt{d}}\right)^{\frac{1}{2} \left(d-2 \sqrt{d} r\right)} \left(\frac{c}{\sqrt{d}}+\frac{1}{2}\right)^{\sqrt{d} r+\frac{d}{2}}(\sqrt{d}-2 c)^{-1} dr \, , \\
    \intertext{Reordering terms and expanding the formulas to power series at $d = \infty$ we get}
    &\sim \frac{1}{2} + \int_{h}^{c} \Big[e^{-2 (r-c)^2} \sqrt{\frac{2}{\pi}} + 2 e^{-2 (r-c)^2} \sqrt{\frac{2}{\pi}} (cp-r) \sqrt{\frac{1}{d}}+ O(1/d)\Big] dr \\
    &= \frac{1}{2} + \frac{1-e^{-2 (r-c)^2}}{\sqrt{2 \pi}} \sqrt{\frac{1}{d}} + \frac{1}{2} \text{Erf}[\sqrt{2} (c-h)] + O(1/d)\, .
\end{align}
For $A_2$ we do not need to integrate, but only use Stirling and do the expansion at $d = \infty$
\begin{align}
    A_2 &= \binom{d-1}{ H - 1} p^{H - 1} (1-p)^{d-H}\\
    &\sim \sqrt{\frac{\bar{d}}{2 \pi (H-1) (\bar{d}-H+1)}} \frac{\bar{d}^{\bar{d}}}{(H-1)^{(H-1)} (\bar{d}-H+1)^{(\bar{d}-H+1)}}p^{H - 1} (1-p)^{d-H} \\
    &\sim\sqrt{\frac{2}{\pi }} e^{-2 (h-c)^2}  \sqrt{\frac{1}{d}}+\frac{2 \sqrt{\frac{2}{\pi }} (h-c) e^{-2 (h-c)^2}}{d}+O((\frac{1}{d})^{3/2})\, .
\end{align}
For the ratio we then get
\begin{align}
    \frac{A_1}{A_2} = \frac{1}{2} \sqrt{\frac{\pi }{2}}  e^{2 (h-c)^2} \text{Erfc}(\sqrt{2} (h-c)) \sqrt{d} + O(\sqrt{\frac{1}{d}})\, .
\end{align}
Inserting this into \eqref{eq:finalfixedddd}, we get the following for $p$
\begin{align}
    p = 1/2 + \frac{e^{-2(h-c)^2}}{\sqrt{2 \pi} \text{Erfc}[\sqrt{2}(h-c)]}\sqrt{\frac{1}{d}} + O(1/d)
\end{align}
so given the assumption \eqref{eq:as1} the factor $c$ is exactly
\begin{align}
    c =  \frac{e^{-2(h-c)^2}}{\sqrt{2 \pi} \text{Erfc}[\sqrt{2}(h-c)]}\, . \label{eq:ceq}
\end{align}
For the entropy we use \eqref{eq:finalent}. For the first term in its sum the approximation can be found analogous to the operations used before to be
\begin{align}
    \log(1/2 - 1/2 \text{Erf}[\sqrt{2} (h-c)]) + \frac{\sqrt{2/\pi}(-1+e^{2(h-c)^2}) e^{-2(h-c)^2}}{\text{Erfc}(\sqrt{2}(h-c))} \sqrt{\frac{1}{d}} + O(1/d)\, .
\end{align}
Finally, the last term can be approximated as
\begin{align}
    d\log(2) - 2 c^2 + 4 c^4 \frac{1}{d} + O(1/d)^{2/3}\, .
\end{align}
Putting this together we get
\begin{align}
    \Phi_{RS} = (-2 c^2+\log(\text{Erfc}(\sqrt{2}(h-c))) +\frac{\sqrt{2/\pi}(-1+e^{2(h-c)^2}) e^{-2(h-c)^2}}{\text{Erfc}(\sqrt{2}(h-c))}  \sqrt{\frac{1}{d}} + O(1/d)\, . \label{eq:lastform}
\end{align}
Using \eqref{eq:ceq} together with \eqref{eq:lastform}, a solution of those equations can be found given that either $h$ or $\Phi_{RS}$ is fixed using numerical methods.

\end{document}

%% file: ferro-comparison-with_algo.tex
\begin{tabular}{llrrrrrr}\toprule   &   &    $s(m=0)$ &    $\Phi_{SP}$ & \multicolumn{2}{l}{EagerSWAP} &\multicolumn{2}{l}{GreedySWAP}  \\$d$ & $H$ &          &          &    $\cdot$ min      &   $\cdot$ median        &    $\cdot$ min      &  $\cdot$ median        \\\midrule
3  & 1 &                                     0.58286 &                                         0.0 &      0.0 &  0.00090 &      0.0 &      0.0 \\
   & 2 &                                     0.23500 &                                         0.0 &  0.00010 &  0.00160 &  0.00050 &  0.00145 \\
\hline 4  & 1 &                                     0.63598 &                                         0.0 &      0.0 &  0.00055 &      0.0 &      0.0 \\
   & 2 &                                     0.41766 &                                         0.0 &      0.0 &  0.00120 &      0.0 &      0.0 \\
   & 3 &  \textit{-0.03179} &  \textit{-0.03179} &  0.07790 &  0.08650 &  0.24240 &  0.24815 \\
\hline 5  & 1 &                                     0.66370 &                                         0.0 &  0.00010 &  0.00105 &      0.0 &      0.0 \\
   & 2 &                                     0.52812 &                                         0.0 &      0.0 &  0.00120 &      0.0 &      0.0 \\
   & 3 &                                     0.21983 &                                         0.0 &  0.00070 &  0.00195 &  0.00040 &  0.00205 \\
   & 4 &  \textit{-0.31188} &  \textit{-0.31188} &  0.50660 &  0.52290 &  0.31980 &  0.32365 \\
\hline 6  & 2 &                                     0.59522 &                                         0.0 &  0.00020 &  0.00150 &      0.0 &      0.0 \\
   & 3 &                                     0.38414 &                                         0.0 &      0.0 &  0.00115 &      0.0 &      0.0 \\
   & 4 &  \textit{-0.00013} &  \textit{-0.00013} &  0.06080 &  0.06385 &  0.18000 &  0.19505 \\
   & 5 &  \textit{-0.60130} &  \textit{-0.60130} &  0.78840 &  0.79760 &  0.56680 &  0.58215 \\
\hline 7  & 2 &                                     0.63586 &                                         0.0 &  0.00020 &  0.00120 &      0.0 &      0.0 \\
   & 3 &                                     0.49266 &                                         0.0 &  0.00020 &  0.00070 &      0.0 &      0.0 \\
   & 4 &                                     0.21369 &                                         0.0 &  0.00040 &  0.00170 &  0.00040 &  0.00150 \\
   & 5 &  \textit{-0.23625} &  \textit{-0.23625} &  0.37230 &  0.38265 &  0.24170 &  0.25465 \\
\hline 8  & 3 &                                     0.56457 &                                         0.0 &  0.00010 &  0.00130 &      0.0 &      0.0 \\
   & 4 &                                     0.36259 &                                         0.0 &  0.00020 &  0.00045 &      0.0 &      0.0 \\
   & 5 &   \textbf{0.02302} &   \textbf{0.02302} &  0.04910 &  0.05280 &  0.15320 &  0.15710 \\
   & 6 &  \textit{-0.48472} &  \textit{-0.48472} &  0.68660 &  0.69545 &  0.48060 &  0.48830 \\
\hline 9  & 3 &                                     0.61193 &                                         0.0 &  0.00010 &  0.00120 &      0.0 &      0.0 \\
   & 4 &                                     0.46712 &                                         0.0 &      0.0 &  0.00085 &      0.0 &      0.0 \\
   & 5 &                                     0.21037 &                                         0.0 &  0.00090 &  0.00110 &  0.00040 &  0.00140 \\
   & 6 &  \textit{-0.18361} &  \textit{-0.18361} &  0.29300 &  0.30285 &  0.20340 &  0.21440 \\
\hline 10 & 4 &                                     0.54044 &                                         0.0 &      0.0 &  0.00090 &      0.0 &      0.0 \\
   & 5 &                                     0.34723 &                                         0.0 &  0.00010 &  0.00085 &      0.0 &      0.0 \\
   & 6 &   \textbf{0.04004} &   \textbf{0.04004} &  0.04010 &  0.04285 &  0.12750 &  0.13760 \\
   & 7 &  \textit{-0.40318} &  \textit{-0.40318} &  0.59110 &  0.60295 &  0.40910 &  0.42655 \\
\hline 11 & 4 &                                     0.59148 &                                         0.0 &  0.00010 &  0.00075 &      0.0 &      0.0 \\
   & 5 &                                     0.44753 &                                         0.0 &      0.0 &  0.00085 &      0.0 &      0.0 \\
   & 6 &                                     0.20829 &                                         0.0 &  0.00060 &  0.00140 &  0.00050 &  0.00235 \\
   & 7 &  \textit{-0.14532} &  \textit{-0.14532} &  0.24490 &  0.24765 &  0.17870 &  0.19435 \\
\hline 12 & 5 &                                     0.52079 &                                         0.0 &  0.00010 &  0.00100 &      0.0 &      0.0 \\
   & 6 &                                     0.33557 &                                         0.0 &  0.00040 &  0.00120 &      0.0 &      0.0 \\
   & 7 &   \textbf{0.05308} &   \textbf{0.05308} &  0.03480 &  0.03760 &  0.11620 &  0.12010 \\
   & 8 &  \textit{-0.34335} &  \textit{-0.34335} &  0.51240 &  0.52095 &  0.36910 &  0.37435 \\
\bottomrule
\end{tabular}

%% file: main.bib
@article{gamarnikMaxcutSparseRandom2018,
  title = {On the Max-Cut of Sparse Random Graphs},
  author = {Gamarnik, David and Li, Quan},
  year = {2018},
  journal = {Random Structures \& Algorithms},
  volume = {52},
  number = {2},
  pages = {219--262},
  issn = {1098-2418},
  doi = {10.1002/rsa.20738},
  langid = {english},
}

@article{zdeborova2011quiet,
  title={Quiet planting in the locked constraint satisfaction problems},
  author={Zdeborov{\'a}, Lenka and Krzakala, Florent},
  journal={SIAM Journal on Discrete Mathematics},
  volume={25},
  number={2},
  pages={750--770},
  year={2011},
  publisher={SIAM}
}

@article{chen2019suboptimality,
  title={Suboptimality of local algorithms for a class of max-cut problems},
  author={Chen, Wei-Kuo and Gamarnik, David and Panchenko, Dmitry and Rahman, Mustazee},
  journal={The Annals of Probability},
  volume={47},
  number={3},
  pages={1587--1618},
  year={2019},
  publisher={Institute of Mathematical Statistics}
}

@article{mora2007geometry,
  title={Geometry and inference in optimization and in information theory},
  author={Mora, Thierry},
  journal={Paris: Universite Pairs Sud-- Pairs XI},
  year={2007}
}

@article{zdeborova2008locked,
  title={Locked constraint satisfaction problems},
  author={Zdeborov{\'a}, Lenka and M{\'e}zard, Marc},
  journal={Physical review letters},
  volume={101},
  number={7},
  pages={078702},
  year={2008},
  publisher={APS}
}

@article{gamarnik2021overlap,
  title={The overlap gap property: A topological barrier to optimizing over random structures},
  author={Gamarnik, David},
  journal={Proceedings of the National Academy of Sciences},
  volume={118},
  number={41},
  year={2021},
  publisher={National Acad Sciences}
}

@book{mezard1987spin,
  title={Spin glass theory and beyond: An Introduction to the Replica Method and Its Applications},
  author={M{\'e}zard, Marc and Parisi, Giorgio and Virasoro, Miguel Angel},
  volume={9},
  year={1987},
  publisher={World Scientific Publishing Company}
}

@article{yedidia2003understanding,
  title={Understanding belief propagation and its generalizations},
  author={Yedidia, Jonathan S and Freeman, William T and Weiss, Yair and others},
  journal={Exploring artificial intelligence in the new millennium},
  volume={8},
  number={236-239},
  pages={0018--9448},
  year={2003}
}

@article{parisi2017marginally,
  title={The marginally stable bethe lattice spin glass revisited},
  author={Parisi, Giorgio},
  journal={Journal of Statistical Physics},
  volume={167},
  number={3},
  pages={515--542},
  year={2017},
  publisher={Springer}
}

@article{muller2015marginal,
  title={Marginal stability in structural, spin, and electron glasses},
  author={M{\"u}ller, Markus and Wyart, Matthieu},
  journal={Annu. Rev. Condens. Matter Phys.},
  volume={6},
  number={1},
  pages={177--200},
  year={2015},
  publisher={Annual Reviews}
}

@article{sherrington1975solvable,
  title={Solvable model of a spin-glass},
  author={Sherrington, David and Kirkpatrick, Scott},
  journal={Physical review letters},
  volume={35},
  number={26},
  pages={1792},
  year={1975},
  publisher={APS}
}

@book{mezard2009information,
  title={Information, physics, and computation},
  author={Mezard, Marc and Montanari, Andrea},
  year={2009},
  publisher={Oxford University Press}
}

@article{alaouiLocalAlgorithmsMaximum2021,
  title = {Local Algorithms for {{Maximum Cut}} and {{Minimum Bisection}} on Locally Treelike Regular Graphs of Large Degree},
  author = {Alaoui, Ahmed El and Montanari, Andrea and Sellke, Mark},
  year = {2021},
  month = nov,
  journal = {arXiv:2111.06813 [math-ph]},
  eprint = {2111.06813},
  eprinttype = {arxiv},
  primaryclass = {math-ph},
  archiveprefix = {arXiv},
  langid = {english},
  
}

@article{angelLocalMaxcutSmoothed2017,
  title = {Local Max-Cut in Smoothed Polynomial Time},
  author = {Angel, Omer and Bubeck, S{\'e}bastien and Peres, Yuval and Wei, Fan},
  year = {2017},
  month = apr,
  journal = {arXiv:1610.04807 [cs, math]},
  eprint = {1610.04807},
  eprinttype = {arxiv},
  primaryclass = {cs, math},
  archiveprefix = {arXiv},
  
}

@article{aubinStorageCapacitySymmetric2019,
  title = {Storage Capacity in Symmetric Binary Perceptrons},
  author = {Aubin, Benjamin and Perkins, Will and Zdeborov{\'a}, Lenka},
  year = {2019},
  month = jun,
  journal = {Journal of Physics A: Mathematical and Theoretical},
  volume = {52},
  number = {29},
  pages = {294003},
  publisher = {{IOP Publishing}},
  issn = {1751-8121},
  doi = {10.1088/1751-8121/ab227a},
  langid = {english},
  
}

@article{banInternalPartitionsRegular2013,
  title = {Internal {{Partitions}} of {{Regular Graphs}}},
  author = {Ban, Amir and Linial, Nati},
  year = {2013},
  month = jul,
  journal = {arXiv:1307.5246 [math]},
  eprint = {1307.5246},
  eprinttype = {arxiv},
  primaryclass = {math},
  archiveprefix = {arXiv},
  langid = {english},
  
}

@inproceedings{bazganExistenceDeterminationSatisfactory2003,
  title = {On the {{Existence}} and {{Determination}} of {{Satisfactory Partitions}} in a {{Graph}}},
  booktitle = {Algorithms and {{Computation}}},
  author = {Bazgan, Cristina and Tuza, Zsolt and Vanderpooten, Daniel},
  editor = {Ibaraki, Toshihide and Katoh, Naoki and Ono, Hirotaka},
  year = {2003},
  series = {Lecture {{Notes}} in {{Computer Science}}},
  pages = {444--453},
  publisher = {{Springer}},
  address = {{Berlin, Heidelberg}},
  doi = {10.1007/978-3-540-24587-2_46},
  isbn = {978-3-540-24587-2},
  langid = {english},
  
}

@article{bazganSatisfactoryGraphPartition2010,
  title = {Satisfactory Graph Partition, Variants, and Generalizations},
  author = {Bazgan, Cristina and Tuza, Zsolt and Vanderpooten, Daniel},
  year = {2010},
  month = oct,
  journal = {European Journal of Operational Research},
  volume = {206},
  number = {2},
  pages = {271--280},
  issn = {03772217},
  doi = {10.1016/j.ejor.2009.10.019},
  langid = {english},
  
}

@article{brayMetastableStatesInternal1981,
  title = {Metastable States, Internal Field Distributions and Magnetic Excitations in Spin Glasses},
  author = {Bray, A. J. and Moore, M. A.},
  year = {1981},
  month = jul,
  journal = {Journal of Physics C: Solid State Physics},
  volume = {14},
  number = {19},
  pages = {2629--2664},
  publisher = {{IOP Publishing}},
  issn = {0022-3719},
  doi = {10.1088/0022-3719/14/19/013},
  langid = {english},
  
}

@article{chenSmoothedComplexityLocal2019,
  title = {Smoothed Complexity of Local {{Max-Cut}} and Binary {{Max-CSP}}},
  author = {Chen, Xi and Guo, Chenghao and {Vlatakis-Gkaragkounis}, Emmanouil-Vasileios and Yannakakis, Mihalis and Zhang, Xinzhi},
  year = {2019},
  month = nov,
  journal = {arXiv:1911.10381 [cs]},
  eprint = {1911.10381},
  eprinttype = {arxiv},
  primaryclass = {cs},
  archiveprefix = {arXiv},
  
}

@article{christopoulosOverviewWhatWe2004,
  title = {An {{Overview}} of {{What We Can}} and {{Cannot Do}} with {{Local Search}}},
  author = {Christopoulos, Petros and Zissimopoulos, Vassilis},
  year = {2004},
  pages = {53},
  langid = {english},
  
}

@article{demboExtremalCutsSparse2017,
  title = {Extremal {{Cuts}} of {{Sparse Random Graphs}}},
  author = {Dembo, Amir and Montanari, Andrea and Sen, Subhabrata},
  year = {2017},
  month = mar,
  journal = {The Annals of Probability},
  volume = {45},
  number = {2},
  eprint = {1503.03923},
  eprinttype = {arxiv},
  issn = {0091-1798},
  doi = {10.1214/15-AOP1084},
  archiveprefix = {arXiv},
  
}

@article{ferberFriendlyBisectionsRandom2021a,
  title = {Friendly Bisections of Random Graphs},
  author = {Ferber, Asaf and Kwan, Matthew and Narayanan, Bhargav and Sah, Ashwin and Sawhney, Mehtaab},
  year = {2021},
  month = jun,
  journal = {arXiv:2105.13337 [math]},
  eprint = {2105.13337},
  eprinttype = {arxiv},
  primaryclass = {math},
  archiveprefix = {arXiv},
  langid = {english},
  
}

@article{gerberAlgorithmicApproachSatisfactory2000,
  title = {Algorithmic Approach to the Satisfactory Graph Partitioning Problem},
  author = {Gerber, Michael U and Kobler, Daniel},
  year = {2000},
  month = sep,
  journal = {European Journal of Operational Research},
  volume = {125},
  number = {2},
  pages = {283--291},
  issn = {0377-2217},
  doi = {10.1016/S0377-2217(99)00459-2},
  langid = {english},
  
}

@article{gomesFindingCutsBounded2021,
  title = {Finding {{Cuts}} of {{Bounded Degree}}: {{Complexity}}, {{FPT}} and {{Exact Algorithms}}, and {{Kernelization}}},
  shorttitle = {Finding {{Cuts}} of {{Bounded Degree}}},
  author = {Gomes, Guilherme C. M. and Sau, Ignasi},
  year = {2021},
  month = jun,
  journal = {Algorithmica},
  volume = {83},
  number = {6},
  pages = {1677--1706},
  issn = {1432-0541},
  doi = {10.1007/s00453-021-00798-8},
  langid = {english},
  
}

@article{hooryExpanderGraphsTheir2006,
  title = {Expander Graphs and Their Applications},
  author = {Hoory, Shlomo and Linial, Nathan and Wigderson, Avi},
  year = {2006},
  month = aug,
  journal = {Bulletin of the American Mathematical Society},
  volume = {43},
  number = {04},
  pages = {439--562},
  issn = {0273-0979},
  doi = {10.1090/S0273-0979-06-01126-8},
  langid = {english},
  
}

@article{hopfieldNeuralNetworksPhysical1982,
  title = {Neural Networks and Physical Systems with Emergent Collective Computational Abilities},
  author = {Hopfield, J. J.},
  year = {1982},
  month = apr,
  journal = {Proceedings of the National Academy of Sciences},
  volume = {79},
  number = {8},
  pages = {2554--2558},
  publisher = {{National Academy of Sciences}},
  issn = {0027-8424, 1091-6490},
  doi = {10.1073/pnas.79.8.2554},
  chapter = {Research Article},
  langid = {english},
  pmid = {6953413},
  
}

@article{krauthStorageCapacityMemory1989,
  title = {Storage Capacity of Memory Networks with Binary Couplings},
  author = {Krauth, Werner and M{\'e}zard, Marc},
  year = {1989},
  journal = {Journal de Physique},
  volume = {50},
  number = {20},
  pages = {3057--3066},
  issn = {0302-0738},
  doi = {10.1051/jphys:0198900500200305700},
  langid = {english},
  
}

@article{kristiansenAlliancesGraphs2004,
  title = {Alliances in Graphs},
  author = {Kristiansen, Petter and Hedetniemi, Sandra and Hedetniemi, Stephen},
  year = {2004},
  month = jan,
  journal = {JCMCC. The Journal of Combinatorial Mathematics and Combinatorial Computing},
  volume = {48}
}

@article{linialAsymptoticallyAlmostEvery2017,
  title = {Asymptotically {{Almost Every}} \$2r\$-Regular {{Graph}} Has an {{Internal Partition}}},
  author = {Linial, Nathan and Louis, Sria},
  year = {2017},
  month = aug,
  journal = {arXiv:1708.04162 [math]},
  eprint = {1708.04162},
  eprinttype = {arxiv},
  primaryclass = {math},
  archiveprefix = {arXiv},
  
}

@article{liuConjectureSchweserStiebitz2021,
  title = {On a Conjecture of {{Schweser}} and {{Stiebitz}}},
  author = {Liu, Muhuo and Xu, Baogang},
  year = {2021},
  month = may,
  journal = {Discrete Applied Mathematics},
  volume = {295},
  pages = {25--31},
  issn = {0166-218X},
  doi = {10.1016/j.dam.2021.02.028},
  langid = {english},
  
}

@article{maDecomposingC4freeGraphs2019,
  title = {Decomposing {{C4-free}} Graphs under Degree Constraints},
  author = {Ma, Jie and Yang, Tianchi},
  year = {2019},
  journal = {Journal of Graph Theory},
  volume = {90},
  number = {1},
  pages = {13--23},
  issn = {1097-0118},
  doi = {10.1002/jgt.22364},
  langid = {english},
  
}

@article{martinFrozenGlassPhase2004,
  title = {Frozen {{Glass Phase}} in the {{Multi-index Matching Problem}}},
  author = {Martin, O. C. and M{\'e}zard, M. and Rivoire, O.},
  year = {2004},
  month = nov,
  journal = {Physical Review Letters},
  volume = {93},
  number = {21},
  pages = {217205},
  publisher = {{American Physical Society}},
  doi = {10.1103/PhysRevLett.93.217205},
  
}

@article{martinRandomMultiindexMatching2005,
  title = {Random Multi-Index Matching Problems},
  author = {Martin, O. C. and M\'ezard, M. and Rivoire, O.},
  year = {2005},
  month = sep,
  journal = {Journal of Statistical Mechanics: Theory and Experiment},
  volume = {2005},
  number = {09},
  eprint = {cond-mat/0507180},
  eprinttype = {arxiv},
  pages = {P09006-P09006},
  issn = {1742-5468},
  doi = {10.1088/1742-5468/2005/09/P09006},
  archiveprefix = {arXiv},
  
}

@article{mckayShortCyclesRandom2004,
  title = {Short {{Cycles}} in {{Random Regular Graphs}}},
  author = {McKay, Brendan D. and Wormald, Nicholas C. and Wysocka, Beata},
  year = {2004},
  month = sep,
  journal = {The Electronic Journal of Combinatorics},
  volume = {11},
  number = {1},
  pages = {R66},
  issn = {1077-8926},
  doi = {10.37236/1819},
  langid = {english},
  
}

@article{morrisContagion2000,
  title = {Contagion},
  author = {Morris, Stephen},
  year = {2000},
  month = jan,
  journal = {The Review of Economic Studies},
  volume = {67},
  number = {1},
  pages = {57--78},
  issn = {0034-6527},
  doi = {10.1111/1467-937X.00121},
  
}

@article{perkinsFrozenRSBStructure2021,
  title = {Frozen \$1\$-{{RSB}} Structure of the Symmetric {{Ising}} Perceptron},
  author = {Perkins, Will and Xu, Changji},
  year = {2021},
  month = feb,
  journal = {arXiv:2102.05163 [math-ph]},
  eprint = {2102.05163},
  eprinttype = {arxiv},
  primaryclass = {math-ph},
  archiveprefix = {arXiv},
  
}

@article{schafferSimpleLocalSearch1991,
  title = {Simple Local Search Problems That Are Hard to Solve},
  author = {Sch{\"a}ffer, Alejandro A. and Yannakakis, Mihalis},
  year = {1991},
  month = feb,
  journal = {SIAM Journal on Computing},
  volume = {20},
  number = {1},
  pages = {56--87},
  issn = {0097-5397},
  doi = {10.1137/0220004}
}

@article{shafiqueSatisfactoryPartitioningGraphs2002,
  title = {On Satisfactory Partitioning of Graphs},
  author = {Shafique, K. and Dutton, R. D.},
  year = {2002},
  journal = {undefined},
  langid = {english},
  
}

@article{stiebitzDecomposingGraphsDegree1996,
  title = {Decomposing Graphs under Degree Constraints},
  author = {Stiebitz, Michael},
  year = {1996},
  journal = {Journal of Graph Theory},
  volume = {23},
  number = {3},
  pages = {321--324},
  issn = {1097-0118},
  doi = {10.1002/(SICI)1097-0118(199611)23:3<321::AID-JGT12>3.0.CO;2-H},
  langid = {english},
  
}

@article{trevesMetastableStatesAsymmetrically1988a,
  title = {Metastable States in Asymmetrically Diluted {{Hopfield}} Networks},
  author = {Treves, A. and Amit, D. J.},
  year = {1988},
  month = jul,
  journal = {Journal of Physics A: Mathematical and General},
  volume = {21},
  number = {14},
  pages = {3155--3169},
  publisher = {{IOP Publishing}},
  issn = {0305-4470},
  doi = {10.1088/0305-4470/21/14/016},
  langid = {english},
  
}

@article{zdeborovaConjectureMaximumCut2010,
  title = {Conjecture on the Maximum Cut and Bisection Width in Random Regular Graphs},
  author = {Zdeborov{\'a}, Lenka and Boettcher, Stefan},
  year = {2010},
  month = feb,
  journal = {Journal of Statistical Mechanics: Theory and Experiment},
  volume = {2010},
  number = {02},
  eprint = {0912.4861},
  eprinttype = {arxiv},
  pages = {P02020},
  issn = {1742-5468},
  doi = {10.1088/1742-5468/2010/02/P02020},
  archiveprefix = {arXiv},
  
}

@article{zdeborovaConstraintSatisfactionProblems2008,
  title={Constraint satisfaction problems with isolated solutions are hard},
  author={Zdeborov{\'a}, Lenka and M{\'e}zard, Marc},
  journal={Journal of Statistical Mechanics: Theory and Experiment},
  volume={2008},
  number={12},
  pages={P12004},
  year={2008},
  publisher={IOP Publishing}
}

@article{zdeborovaStatisticalPhysicsHard2008a,
  title = {Statistical {{Physics}} of {{Hard Optimization Problems}}},
  author = {Zdeborov{\'a}, Lenka},
  year = {2008},
  month = jun,
  journal = {arXiv:0806.4112 [cond-mat]},
  eprint = {0806.4112},
  eprinttype = {arxiv},
  primaryclass = {cond-mat},
  archiveprefix = {arXiv},
  
}
